\documentclass[11pt, draftcls, onecolumn]{IEEEtran}
\IEEEoverridecommandlockouts 
\usepackage{makecell}
\usepackage{bbm}
\usepackage{graphicx}
\usepackage{subfigure}
\usepackage[usenames,dvipsnames]{color}
\usepackage{epsfig}
\usepackage{amssymb}
\usepackage{amsmath}
\usepackage{amsthm}
\usepackage{latexsym}
\usepackage{setspace}
\usepackage{bbm}
\usepackage{flushend}
\usepackage[top=1in, bottom=1in, left=1in, right=1in]{geometry}
\usepackage{float}
\usepackage{dsfont}
\usepackage{caption}
\usepackage[ruled,vlined]{algorithm2e}
\usepackage{xspace}

\newcommand{\lc}{\left(}
\newcommand{\rc}{\right)}
\newcommand{\ls}{\left[}
\newcommand{\rs}{\right]}

\newcommand{\dP}{\mathrm{P}}
\newcommand{\dQ}{\mathrm{Q}}
\newcommand{\bP}[2]{\mathrm{P}_{#1}\left(#2\right)}
\newcommand{\bQ}[2]{\mathrm{Q}_{#1}\left(#2\right)}
\newcommand{\bPP}[1]{\mathrm{P}_{#1}}
\newcommand{\bQQ}[1]{\mathrm{Q}_{#1}}
\newcommand{\bPr}[1]{{\mathrm{Pr}}\left(#1\right)}

\newcommand{\bE}[2]{{\mathbb{E}}_{#1}\left\{{#2}\right\}}
\newcommand{\bEE}[1]{{\mathbb{E}}\left[#1\right]}

\newcommand{\cA}{{\mathcal A}}
\newcommand{\cB}{{\mathcal B}}
\newcommand{\cC}{{\mathcal C}}

\newcommand{\cE}{{\mathcal E}}
\newcommand{\cF}{{\mathcal F}}

\newcommand{\cJ}{{\mathcal J}}
\newcommand{\cK}{{\mathcal K}}

\newcommand{\mN}{{\mathbbm N}}

\newcommand{\cO}{{\mathcal O}}

\newcommand{\cT}{{\mathcal T}}
\newcommand{\cU}{{\mathcal U}}
\newcommand{\cV}{{\mathcal V}}

\newcommand{\cX}{{\mathcal X}}
\newcommand{\cY}{{\mathcal Y}}
\newcommand{\cZ}{{\mathcal Z}}

\makeatletter
\newtheorem*{rep@theorem}{\rep@title}
\newcommand{\newreptheorem}[2]{%
\newenvironment{rep#1}[1]{%
 \def\rep@title{#2 \ref{##1}}%
 \begin{rep@theorem}}%
 {\end{rep@theorem}}}
\makeatother

\newtheorem{theorem}{Theorem}

\newtheorem{corollary}[theorem]{Corollary}
\newtheorem*{corollary*}{Corollary}

\newtheorem{lemma}[theorem]{Lemma}
\newtheorem*{lemma*}{Lemma}

\newreptheorem{theorem}{Theorem}
\newreptheorem{lemma}{Lemma}

\theoremstyle{remark}
\newtheorem{remark}{Remark}
\newtheorem*{remark*}{Remark}
\newtheorem*{remarks*}{Remarks}

\theoremstyle{definition}
\newtheorem{definition}{Definition}
\newtheorem{example}{Example}

\newcommand{\ed}{\stackrel{{\rm def}}{=}}

\def\undertilde#1{\mathord{\vtop{\ialign{##\crcr
$\hfil\displaystyle{#1}\hfil$\crcr\noalign{\kern1.5pt\nointerlineskip}
$\hfil\tilde{}\hfil$\crcr\noalign{\kern1.5pt}}}}}

\newcommand{\ep}{\varepsilon}
\newcommand{\la}{\lambda}

\newcommand{\ttlvrn}[2]{d_{\mathtt{var}}\left( #1 , #2\right)}

\newcommand{\indicator}{{\mathds{1}}}
\newcommand{\mc}{-\!\!\!\!\circ\!\!\!\!-}

\newcommand{\ie}{$i.e.$}
\newcommand{\cf}{$cf.$~}
\newcommand{\etc}{$etc.$\xspace}

\allowdisplaybreaks
\newcommand{\prot}{\pi}
\newcommand{\Prot}{\Pi}
\newcommand{\protr}{\tau}
\newcommand{\Pprot}{\Pi \Pi}
\newcommand{\simm}{\mathtt{sim}}
\newcommand{\DE}{\mathtt{DE}}
\newcommand{\psim}{\prot_{\simm}} 
\newcommand{\Psim}{\Prot_{\simm} }
\newcommand{\pcom}{\prot_{\mathtt{com}} }

\newcommand{\pdex}{\prot_{\DE}}
\newcommand{\Pdex}{\Prot_{\DE}}
\newcommand{\pmix}{\prot_{\mathtt{mix}}}
\newcommand{\pmixn}{\prot_{\mathtt{mix,n}}}
\newcommand{\Protx}{\Prot_{\cX}}
\newcommand{\Proty}{\Prot_{\cY}}

\newcommand{\dc}[1]{D_\ep\left(#1\right)}
\newcommand{\dcp}{\dc {\prot}}
\newcommand{\dcpp}{\dc {\prot| \bPP{XY}}}

\newcommand{\rmax}{r_{\max}}
\newcommand{\ic}{\mathtt{ic}}
\newcommand{\IC}{\mathtt{IC}}
\newcommand{\icpp}{\ic(\protr; x,y)}
\newcommand{\icp}{\ic(\Prot; X,Y)}
\newcommand{\Icp}{\IC(\prot)}
\newcommand{\varp}{\mathtt{V}(\prot)}

\newcommand{\hxy}{h(X,Y)}
\newcommand{\hsum}[2]{h\left(#1 \triangle #2\right)}
\newcommand{\hsump}{\hsum {X\Prot}{Y\Prot}}
\newcommand{\lamin}{\lambda_{\min}}
\newcommand{\lamax}{\lambda_{\max}}
\newcommand{\lamins}[1]{\lamin^{(#1)}}
\newcommand{\lamaxs}[1]{\lamax^{(#1)}}
\newcommand{\den}{\zeta}

\newcommand{\epsmall}{\ep^\prime}
\newcommand{\eptail}{\ep_{\mathtt{tail}}}
\newcommand{\lasmall}{\la^\prime}

\newcommand{\Hmin}{H_{\min}}
\newcommand{\supp}{\mathtt{supp}}

\newcommand{\Esim}{\cE_{\simm}}
\newcommand{\EDE}{\cE_{\DE}}

\newcommand{\partyx}{Party 1\xspace}
\newcommand{\partyy}{Party 2\xspace}

\newcommand{\prots}{\boldsymbol{\prot}}
\newcommand{\sIcp}{\overline{\mathtt{IC}}(\boldsymbol{\prot}) }

\newcommand{\Order}{\cO}
\newcommand{\order}{o}
\newcommand{\tOrder}{\tilde{\cO}}

\newcommand{\hashfunc}{f}
\newcommand{\Hash}{\mathsf{hash}}
\newcommand{\hash}{{\cal F}}

\newcommand{\Sim}{\mathsf{sim}}
\newcommand{\hQxy}{h_{\bQQ{X|Y}}(x|y)}
\newcommand{\hQXY}{h_{\bQQ{X|Y}}(X|Y)}
\newcommand{\hPxy}{h_{\bPP{X|Y}}(x|y)}
\newcommand{\hPXY}{h_{\bPP{X|Y}}(X|Y)}
\newcommand{\lambdaQXYmin}{\lambda_{\bQQ{X|Y}}^{\min}}
\newcommand{\lambdaQXYmax}{\lambda_{\bQQ{X|Y}}^{\max}}
\newcommand{\DeltaQXY}{\Delta_{\bQQ{X|Y}}} 
\newcommand{\NQXY}{N_{\bQQ{X|Y}}}
\newcommand{\lambdaQXYi}{\lambda_{\bQQ{X|Y}}^{(i)}}
\newcommand{\lambdaQXYj}{\lambda_{\bQQ{X|Y}}^{(j)}}
\newcommand{\TQXYzero}{\cT_{\bQQ{X|Y}}^{(0)}}
\newcommand{\TQXYi}{\cT_{\bQQ{X|Y}}^{(i)}}
\newcommand{\TQXYj}{\cT_{\bQQ{X|Y}}^{(j)}}
\newcommand{\TQXYone}{\cT_{\bQQ{X|Y}}^{(1)}}
\newcommand{\TPXY}{\cT_{\bPP{X|Y}}}
\newcommand{\TQXY}{\cT_{\bQQ{X|Y}}}
\newcommand{\DeltaQPiOneY}{\Delta_{\bQQ{\Prot_1|Y}}}
\newcommand{\NQPiOneY}{N_{\bQQ{\Prot_1|Y}}}
\newcommand{\hQPiOneYPioneX}{h_{\bQQ{\Prot_1|Y}}(\Prot_{1\cX}|Y)}
\newcommand{\TQPiOneYzero}{\cT_{\bQQ{\Prot_1|Y}}^{(0)}}
\newcommand{\lambdaQPiOneYmin}{\lambda_{\bQQ{\Prot_1|Y}}^{\min}}
\newcommand{\DeltaPPiOneY}{\Delta_{\bPP{\Prot_1|Y}}}
\newcommand{\NPPiOneY}{N_{\bPP{\Prot_1|Y}}}
\newcommand{\hPPiOneYPioneX}{h_{\bPP{\Prot_1|Y}}(\Prot_{1\cX}|Y)}
\newcommand{\TPPiOneYzero}{\cT_{\bPP{\Prot_1|Y}}^{(0)}}
\newcommand{\lambdaPPiOneYmin}{\lambda_{\bPP{\Prot_1|Y}}^{\min}}
\newcommand{\DeltaPPiOneX}{\Delta_{\bPP{\Prot_1|X}}}
\newcommand{\NPPiOneX}{N_{\bPP{\Prot_1|X}}}
\newcommand{\TPPiOneXzero}{\cT_{\bPP{\Prot_1|X}}^{(0)}}
\newcommand{\TPPiOneXj}{\cT_{\bPP{\Prot_1|X}}^{(j)}}
\newcommand{\lambdaPPiOneXmin}{\lambda_{\bPP{\Prot_1|X}}^{\min}}
\newcommand{\hPPiOneXPioneX}{h_{\bPP{\Prot_1|X}}(\Prot_{1\cX}|X)}
\newcommand{\hPPiOneXPione}{h_{\bPP{\Prot_1|X}}(\Prot_1|X)}
\newcommand{\lambdaPPitXmin}{\lambda_{\bPP{\Prot_t|X \Prot^{t-1}}}^{\min}}
\newcommand{\lambdaPPitYmin}{\lambda_{\bPP{\Prot_t|Y \Prot^{t-1}}}^{\min}}

\newcommand{\DeltaPPitX}{\Delta_{\bPP{\Prot_t|X\Prot^{t-1}}}}
\newcommand{\DeltaPPitY}{\Delta_{\bPP{\Prot_t|Y\Prot^{t-1}}}}
\newcommand{\NPPitX}{N_{\bPP{\Prot_t|X\Prot^{t-1}}}}
\newcommand{\NPPitY}{N_{\bPP{\Prot_t|Y\Prot^{t-1}}}}
\newcommand{\TPPitXzero}{\cT_{\bPP{\Prot_t|X\Prot^{t-1}}}^{(0)}}
\newcommand{\TPPitYzero}{\cT_{\bPP{\Prot_t|Y\Prot^{t-1}}}^{(0)}}

\newenvironment{protocol}[1][htb]
  {%
   \begin{algorithm}[#1]%
  }{\end{algorithm}}

\begin{document}

\title{Information Complexity Density and Simulation of
  Protocols}

\author{
\IEEEauthorblockN{Himanshu Tyagi$^\ast$} 
\and
\IEEEauthorblockN{Shaileshh Venkatakrishnan$^\dag$} 
\and
\IEEEauthorblockN{Pramod
  Viswanath$^\dag$} 
\IEEEauthorblockN{Shun Watanabe$^\ddag$} 
}

\maketitle

{\renewcommand{\thefootnote}{}\footnotetext{
%\hspace*{-.11in}\rule{24ex}{.05em}
\noindent$\ast$Department of Electrical Communication Engineering, Indian
Institute of Science, Bangalore 560012, India. 
Email: htyagi@ece.iisc.ernet.in

\noindent$\ast$Department of Electrical and Computer Engineering, University of
  Illinois, Urbana-Champaign, IL 61801, USA.
Email: \{bjjvnkt2, pramodv\}@illinois.edu

\noindent$^\ddag$Department of Computer and Information Sciences, 
Tokyo University of Agriculture and Technology, Tokyo 184-8588, Japan. 
Email: shunwata@cc.tuat.ac.jp

A preliminary version of this paper was presented at the 7th Innovations in Theoretical Computer Science (ITCS) conference, Cambridge, Massachusetts, USA, 2016.
}}

\maketitle

\renewcommand{\thefootnote}{\arabic{footnote}}
\setcounter{footnote}{0}

\begin{abstract}
Two parties observing correlated random variables seek to run an
interactive communication protocol. How many bits must they
exchange to simulate the protocol, namely to produce a view with 
a joint distribution within a fixed statistical distance of the joint
distribution of the input and the transcript of the original protocol?
We present an information spectrum   
approach for this problem whereby the information complexity of the
protocol is replaced by its information complexity density. Our
single-shot bounds relate the communication complexity of simulating a
protocol to tail bounds for information complexity density. As a
consequence, we obtain a strong converse and characterize the second-order asymptotic
term in communication complexity for independent and identically
distributed observation sequences. Furthermore, we obtain a general
formula for the rate of communication complexity which applies to any
sequence of observations and protocols. Connections with results from
theoretical computer science and implications for the function
computation problem are discussed.
\end{abstract}

%%%%%%%%%%%%%%%%%%%%%%%%%%%%%%%%%%%%%%%%%%%%%%%
\section{Introduction}\label{s:introduction}
Two parties observing random variables $X$ and $Y$ seek to
run an interactive protocol $\prot$ with inputs $X$ and $Y$.
The parties have access to private as well as shared public
randomness. What is the minimum number of bits that they
must exchange in order to simulate $\prot$ to within a fixed
statistical distance $\ep$?  This question is of importance
to the theoretical computer science as well as the
information theory communities. On the one hand, it is
related closely to the communication complexity problem
\cite{Yao79}, which in turn is an important tool for
deriving lower bounds for computational complexity
\cite{KarchmerW88} and for space complexity of streaming
algorithms \cite{AlonMS96}. On the other hand, it is a
significant generalization of the classical information
theoretic problem of distributed data compression
\cite{SleWol73}, replacing data to be compressed with an
interactive protocol and allowing interactive communication
as opposed to the usual one-sided communication.

In recent years, it has been argued that the distributional
communication complexity for simulating a protocol\footnote{The
  difference  between simulation and compression of protocols is
  significant and is discussed in Remark~\ref{simulation_compression}
  below.} $\prot$  
is related closely to its {\it information
  complexity}\footnote{For brevity, we do not display the
  dependence of $\Icp$ on the (fixed) distribution
  $\bPP{XY}$.} $\Icp$ defined as follows:
\[
\Icp \ed I(\Prot \wedge X|Y) + I(\Prot \wedge Y|X),
\]
where $I(X \wedge Y|Z)$ denotes the conditional mutual
information between $X$ and $Y$ given $Z$ (\cf \cite{Sha48,
  CsiKor11}).  For a protocol $\prot$ with communication
complexity $|\prot|$ (the depth of the binary protocol
tree), a simulation protocol requiring $\tOrder(\sqrt{\Icp
  |\prot|})$ bits of communication was given in
\cite{BarakBCR10} and one requiring $2^{\Order(\Icp)}$ bits
of communication was given in
\cite{Braverman12}. A general version
of the simulation problem was considered in
\cite{YasGohAre12}, but only bounded round simulation
protocols were considered. Interestingly, it was shown in
\cite{BraRao11} that the amortized\footnote{Throughout the paper, "amortized" indicates that
the observations are  independently identically distributed (IID) and the protocol to be simulated is
$n$ copies of the same protocol.} 
distributional communication complexity of simulating $n$ copies of a
protocol $\prot$ for vanishing simulation error is bounded
above by\footnote{Braverman and Rao actually used their
  general simulation protocol as a tool for deriving the
  amortized distributional communication complexity of
  function computation. This result was obtained
  independently by Ma and Ishwar in \cite{MaIsh11} using
  standard information theoretic techniques.}  $\Icp$. 
While
a matching lower bound was also derived in \cite{BraRao11},
it is not valid in our context -- \cite{BraRao11} considered
function computation and used a coordinate-wise error
criterion. Nevertheless, we can readily modify the lower
bound argument in \cite{BraRao11} and use the continuity of
conditional mutual information to formally obtain the
required lower bound and thereby a
characterization of the amortized distributional
communication complexity for vanishing simulation
error. Specifically, denoting by $D(\prot^n)$ the
distributional communication complexity of simulating $n$
copies of a protocol $\prot$ with vanishing simulation
error, we have
\begin{align} \nonumber
\lim_{n \rightarrow \infty}\frac 1n D(\prot^n)= \Icp.
\end{align}
Perhaps motivated by this characterization, or a folklore
version of it, the research in this area has focused on
designing simulation protocols for $\prot$ requiring
communication of length depending on $\Icp$; the results
cited above belong to this category as well. However, the
central role of $\Icp$ in the distributional communication
complexity of protocol simulation is far from settled and
many important questions remain unanswered. For instance,
(a) does $\Icp$ suffice to capture the dependence of
distributional communication complexity on the simulation
error $\ep$? (b) Does information complexity have an
operational role in simulating $\prot^n$ besides being the
leading asymptotic term?  (c) How about the simulation of
more complicated protocols such as a mixture $\pmix$ of two
product protocols $\prot_1^n$ and $\prot_2^{n}$ -- does
$\IC(\pmix)$ still constitute the leading asymptotic term in
the communication complexity of simulating $\pmix$?

The quantity $\IC(\prot)$ plays the same role in the simulation
of protocols as $H(X)$ in the compression of $X^n$
\cite{Sha48} and $H(X|Y)$ in the transmission of $X^n$ by
the first to the second party with access to
$Y^n$~\cite{SleWol73}. The questions raised above have been
addressed for these classical problems ($cf.$~\cite{Han03}).
In this paper, we answer these questions for simulation of
interactive protocols. In particular, we answer all these
questions in the negative
by exhibiting another quantity that plays such a fundamental
role and can differ from information complexity
significantly. To this end, we introduce the notion of {\it
  information complexity density} of a protocol $\prot$ with
inputs $X$ and $Y$ generated from a fixed distribution
$\bPP{XY}$.
\begin{definition}[{\bf Information complexity density}] \label{d:ICD}
The {\it information complexity density} of a private coin
protocol $\prot$ is given by the function
\[
\icpp = \log
\frac{\bP{\Prot|XY}{\protr|x,y}}{\bP{\Prot|X}{\protr|x}} +
\log
\frac{\bP{\Prot|XY}{\protr|x,y}}{\bP{\Prot|Y}{\protr|y}},
\] 
for all observations $x$ and $y$ of the two parties and all
transcripts $\protr$, where $\bPP{\Prot XY}$ denotes the
joint distribution of the observation of the two parties and
the random transcript $\Prot$ generated by $\prot$.
\end{definition}
Note that $\Icp = \bEE{\icp}$. We show that it is the
$\ep$-{\it tail of the information complexity density}
$\icp$, \ie, the supremum\footnote{Formally, our lower bound
  uses lower $\ep$-tail $\sup\{\la: \bPr{\icp >\la} > \ep\}$
  and the upper bound uses upper $\ep$-tail $\inf\{\la:
  \bPr{\icp>\la} < \ep\}$. For many interesting cases, the
  two coincide.} over values of $\la$ such that $\bPr{\icp >
  \la} > \ep$, which governs the communication complexity
of simulating a protocol with simulation error less than
$\ep$ and not the information complexity of the
protocol. The information complexity $\Icp$ becomes the
leading term in communication complexity for simulating
$\prot$ only when roughly
\[
\Icp \gg \sqrt{\mathrm{Var}(\icp)\log (1/\ep)}.
\]
This condition holds, for instance, in the amortized regime
considered in \cite{BraRao11}. However, the $\ep$-tail of
$\icp$ can differ significantly from $\Icp$, the mean of
$\icp$.  In Appendix \ref{appendix:small-error}, we provide
an example protocol with inputs of size $2^n$ such that for
$\ep = 1/n^{3}$, the $\ep$-tail of $\icp$ is greater than
$2n$ while $\Icp$ is very small, just $\tOrder( n^{-2})$.

%%%%%%%%%%%%%%%%%%%%
\subsection{Summary of results} 
Our main results are bounds for distributional communication
complexity $\dcp$ for $\ep$-simulating a protocol $\prot$.
The key quantity in our bounds is the $\ep$-tail $\la_\ep$
of $\icp$.

{\bf Lower bound.} Our main contribution is a general lower
bound for $\dcp$.  We show that for every private coin
protocol $\prot$, $\dcp \gtrsim \la_\ep$. In fact, this
bound does not rely on the structure of random variable
$\Prot$ and is valid for the more general problem of
simulating a correlated random variable.

Prior to this work, there was no lower bound that captured
both the dependence on simulation error $\ep$ as well as the
underlying probability distribution. On the one hand, the
lower bound above yields many sharp results in the amortized
regime.  It gives the leading asymptotic term in the
communication complexity for simulating any sequence of
protocols, and not just product protocols.  For product
protocols, it yields the precise dependence of communication
complexity on $\ep$ as well as the exact second-order
asymptotic term. On the other hand, it sheds light on the
dependence of $\dcp$ on $\ep$ even in the single-shot
regime. For instance, our lower bound can be used to exhibit 
an arbitrary separation between $\dcp$ and $\Icp$ when $\ep$
is not fixed. Specifically,
consider the example
protocol in Appendix \ref{appendix:small-error}. On
evaluating our lower bound for this protocol, for
$\ep=1/n^3$ we get $\dcp = \Omega(n)$ which is far more than
$2^{\Icp}$ since $\Icp = \tOrder(n^{-2})$.  Remarkably,
\cite{GanorKR14, GanorKR14ii} exhibited exponential
separation between the distributional communication
complexity of computing a function and the information
complexity of that function even for a fixed $\ep$,
thereby establishing the optimality of the upper bound
$\dcp \leq \Order(2^\Icp)$ given in \cite{Braverman12}.
Our
simple example shows a much stronger separation between
$\dcp$ and $\Icp$, albeit for a vanishing $\ep$.

{\bf Upper bound.} To establish our asymptotic results, we
propose a new simulation protocol, which is of independent
interest. For a protocol $\prot$ with bounded rounds of
interaction, using our proposed protocol we can show that
$\dcp \lesssim \la_\ep$. Much as the protocol of
\cite{BraRao11}, our simulation protocol simulates one round
at a time, and thus, the slack in our upper bound does
depend on the number of rounds.

Note that while the operative term in the lower bound and
the upper bound is the $\ep$-tail of $\icp$, the lower bound
approaches it from below and the upper bound approaches it
from above. It is often the case that these two limits match
and the leading term in our bounds coincide. See
Figure~\ref{f:tails} for an illustration of our bounds.
\begin{figure}[h]
\begin{center}
\includegraphics[scale=0.3]{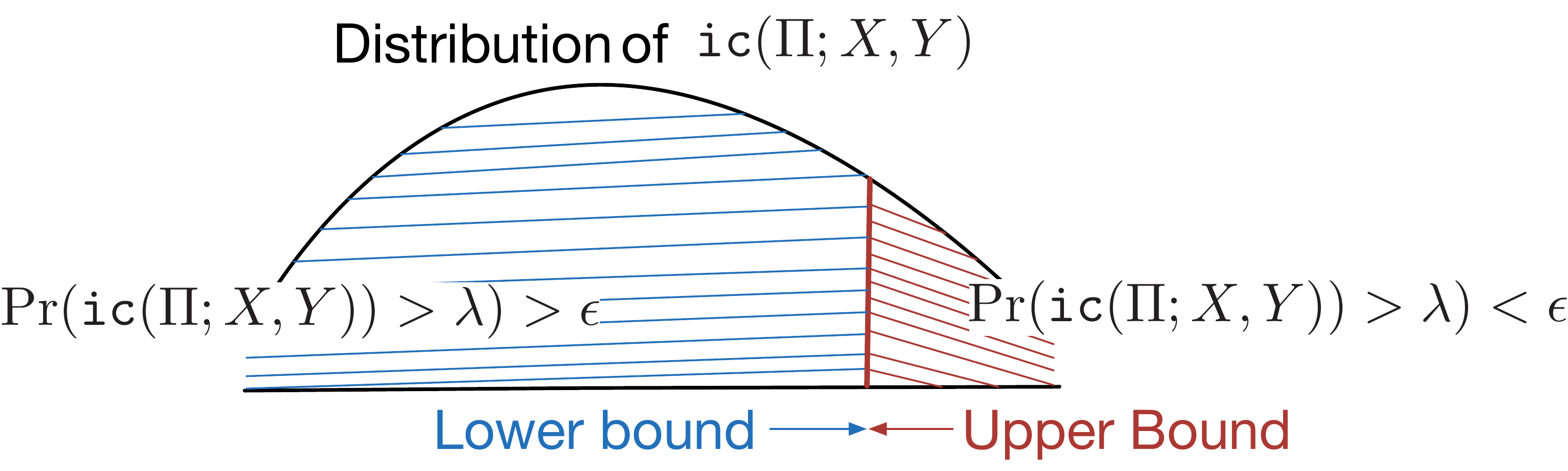}
\caption{Illustration of lower and upper bounds for $\dcp$}
\label{f:tails}
\end{center}
\end{figure}

{\bf Amortized regime: second-order asymptotics.} Denote by
$\prot^n$ the $n$-fold product protocol obtained by applying
$\prot$ to each coordinate $(X_i, Y_i)$ for inputs $X^n$ and
$Y^n$.  Consider the communication complexity
$D_\ep(\prot^n)$ of $\ep$-simulating $\prot^n$ for {\it
  independent and identically distributed} (IID) $(X^n,
Y^n)$ generated from $\bPP{XY}^n$.  Using the bounds above,
we can obtain the following sharpening of the results of
\cite{BraRao11}: With $\mathtt{V}(\prot)$ denoting the
variance of $\icp$,
\[
D_\ep(\prot^n) = n \Icp + \sqrt{n\mathtt{V}(\prot)
}Q^{-1}(\ep) + \order(\sqrt{n}),
\]
where $Q(x)$ is equal to the probability that a standard
normal random variable exceeds $x$ and $Q^{-1}(\ep) \approx
\sqrt{\log (1/\ep)}$. On the other hand, the arguments
in\footnote{The proof in \cite{BraRao11} uses the inequality
  $\Icp \leq |\prot|$, a multiparty extension of which is
  available in \cite{CsiNar08, MadTet10}.} \cite{BraRao11}
or \cite{YasGohAre12} give us
\[
D_\ep(\prot^n)\geq n\Icp - n\ep[|\prot| + \log|\cX||\cY|]-
\ep\log(1/\ep).
\] 
But the precise communication requirement is not less but
$\sqrt {n\mathtt{V}(\prot)\log(1/\ep)}$ {\it more than} $n
\Icp$.

{\bf General formula for amortized communication
  complexity.} The lower and upper bounds above can be used
to derive a formula for the first-order asymptotic term, the
coefficient of $n$, in $D_\ep(\prot_n)$ for any sequence of
protocols $\prot_n$ with inputs $X_n\in \cX^n$ and $Y_n\in
\cY^n$ generated from any sequence of distributions
$\bPP{X_nY_n}$. We illustrate our result by the following
example.
\begin{example}[{\bf Mixed protocol}]\label{ex:mixed_protocol}
Consider two protocols $\prot_\mathtt{h}$ and
$\prot_\mathtt{t}$ with inputs $X$ and $Y$ such that
$\mathtt{IC}(\prot_\mathtt{h}) >
\mathtt{IC}(\prot_\mathtt{t})$.  For $n$ IID observations
$(X^n, Y^n)$ drawn from $\bPP{XY}$, we seek to simulate the
mixed protocol $\pmixn$ defined as follows: \partyx first
flips a (private) coin with probability $p$ of heads and
sends the outcome $\Prot_0$ to \partyy.  Depending on the
outcome of the coin, the parties execute $\prot_\mathtt{h}$
or $\prot_\mathtt{t}$ $n$ times, i.e., they use
$\prot_\mathtt{h}^n$ if $\Prot_0 = \mathtt{h}$ and
$\prot_\mathtt{t}^n$ if $\Prot_0 = \mathtt{t}$.  What is the
amortized communication complexity of simulating the mixed
protocol $\pmixn$? Note that
\[
\mathtt{IC}(\pmixn) = n\ls p\mathtt{IC}(\prot_\mathtt{h}) +
(1-p) \mathtt{IC}(\prot_\mathtt{t})\rs.
\]
Is it true that in the manner of \cite{BraRao11} the leading
asymptotic term in $D_\ep(\pmixn)$ is $\IC(\pmixn)$? In
fact, it is not so. Our general formula implies that for all
$p\in(0,1)$,
\[
D_\ep(\pmixn) = n\mathtt{IC}(\prot_\mathtt{h}) + \order(n)
\] 
This is particularly interesting when $p$ is very small and
$\mathtt{IC}(\prot_\mathtt{h})\gg
\mathtt{IC}(\prot_\mathtt{t})$.
\end{example}
%The results above illustrate the central thesis of this
%paper: It is the $\ep$-tail of $\icp$ and not just $\Icp$
%that governs the communication complexity of
%$\ep$-simulating a protocol $\prot$.
%To derive the $\ep$-tail characterization, 
%existing techniques that manipulate 
%average quantities such as
%the mutual information do not work, and 
%new techniques are needed, which will be described next. 
 
%%%%%%%%%%%%%%%%%%%%
\subsection{Proof techniques}

{\bf Proof for the lower bound.} We present a new method for
deriving lower bounds on distributional communication
complexity. Our proof relies on a reduction argument that
utilizes an $\ep$-simulation to generate an information
theoretically secure secret key for $X$ and $Y$ (for a
definition of the latter, see \cite{Mau93,AhlCsi93} or
Section~\ref{s:background}).  Heuristically, a protocol can
be simulated using fewer bits of communication than its
length because of the correlation in $X$
and $Y$. Due to this correlation, when simulating the
protocol, the parties agree on more bits (generate more {\it
  common randomness}) than what they communicate.  These
extra bits can be extracted as an information theoretically
secure secret key for the two parties using the {\it
  leftover hash lemma} (\cf \cite{BenBraCreMau95,
  RenWol05}).  A lower bound on the number of bits
communicated can be derived using an upper bound for the
maximum possible length of a secret key that can be
generated using interactive communication; the latter was
derived recently in \cite{TyaWat14, TyaWat14ii}.

{\bf Protocol for the upper bound.}  We simulate a given
protocol one round at a time.  Simulation of each round
consists of two subroutines: Interactive Slepian-Wolf
compression and message reduction by public randomness.  The
first subroutine is an interactive version of the classical
Slepian-Wolf compression \cite{SleWol73} for sending $X$ to
an observer of $Y$ which is of optimal instantaneous
rate. The second subroutine uses an idea that appeared first
in \cite{RenRen11} (see, also, \cite{Mur14, YasAreGoh14})
and reduces the number of bits communicated in the first by
realizing a portion of the required communication by the
shared public randomness.  This is possible since we are not
required to recover a given random variable $\Prot$, but
only simulate it to within a fixed statistical distance.

The proposed protocol is closely related to that in
\cite{BraRao11}.  However, there are some crucial
differences. The protocol in \cite{BraRao11}, too, uses
public randomness to sample each round of the protocol,
before transmitting it using an interactive communication of
size incremented in steps.  However, our information
theoretic approach provides a systematic method for choosing
this step size. Furthermore, our protocol for sampling the
protocol from public randomness is significantly different
from that in \cite{BraRao11} and relies on randomness
extraction techniques. In particular, the protocol in
\cite{BraRao11} does not attain the asymptotically optimal
bounds achieved by our protocol.

{\bf Technical approach.} While we utilize new, bespoke
techniques for deriving our lower and upper bounds, casting
our problem in an information theoretic framework allows us
to build upon the developments in this classic field.  In
particular, we rely on the {\it information spectrum
  approach} of Han and Verd\'u, introduced in the seminal
paper \cite{HanVer93} (see the textbook \cite{Han03} for a
detailed account).  In this approach, the classical measures
of information such as entropy and mutual information are
viewed as expectations of the corresponding {\it information
  densities}, and the notion of ``typical sets'' is replaced
by sets where these information densities are bounded
uniformly.  The distribution of an information density (such as $h(x)
= -\log\bP X x$), or the support of this distribution, is loosely
referred to as its {\it spectrum}. Further, we shall refer to the
difference between $\max$ and $\min$ value of $h(x)$ over its support
as the {\it length of the spectrum}. Coding theorems of classical information theory
consider IID repetitions and rely on the so-called the {\it
  asymptotic equipartition property} \cite{CovTho06} which
essentially corresponds to the concentration of spectrums on
small intervals. For {\it single-shot} problems such
concentrations are not available and we have to work with
the whole span of the spectrum.

Our main technical contribution in this paper is the
extension of the information spectrum method to handle
interactive communication.  Our results rely on the analysis 
of appropriately chosen information densities and, in
particular, rely on the spectrum of the information
complexity density $\icp$. Different
components of our analysis require bounds on these
information densities in different directions, which in turn
renders our bounds loose and incurs a gap equal to the
length of the corresponding information spectrum.  To
overcome this shortcoming, we use the {\it spectrum slicing}
technique of Han \cite{Han03}\footnote{The spectrum slicing
  technique was introduced in \cite{Han03} to derive the
  error exponents of various problems for general sources
  and a rate-distortion function for general sources.}  to
divide the information spectrum into small portions with
information densities closely bounded from both sides.
While in our upper bounds spectrum slicing is used to
carefully choose the parameters of the protocol, it is
required in our lower bounds to identify a set of inputs
where a given simulation will require a large number of bits
to be communicated.

%%%%%%%%%%%%%%%%%%%%
\subsection{Organization}
A formal statement of the problem along with the necessary
preliminaries is given in the next section.
Section~\ref{s:main_results} contains all our results.  In
Section~\ref{s:background}, we review the information
theoretic secret key agreement problem, the leftover hash
lemma, and the data exchange problem, all of which will be
instrumental in our proofs. The formal proof of our lower
bound is contained in Section~\ref{s:lower_bound_proof} and
that of our upper bound in
Section~\ref{s:simulation_protocols}.  
Section \ref{sec:asymptotic-optimality} 
contains a proof of our asymptotic results, followed by concluding
remarks in Section \ref{s:conclusion}.

%%%%%%%%%%%%%%%%%%%%
\subsection{Notations}
Random variables are denoted by capital letters such as $X$,
$Y$, \etc realizations by small letters such as $x$, $y$,
\etc and their range sets by corresponding calligraphic
letters such as $\cX$, $\cY$, \etc.  Protocols are denoted
by appropriate subscripts or superscripts with $\prot$, the
corresponding random transcripts by the same sub- or
superscripts with $\Prot$; $\tau$ is used as a placeholder
for realizations of random transcripts. All the logarithms
in this paper are to the base $2$.

The following convention, described for the entropy density,
shall be used for all information densities used in this
paper. We shall abbreviate the entropy density $h_{\bPP
  X}(x) = -\log \bP X x$ by $h(x)$, when there is no
confusion about $\bPP X$, and the random variable $h(X)$
corresponds to drawing $X$ from the distribution $\bPP X$.

Whenever there is no confusion, we will not display the
dependence of distributional communication complexity on the
underlying distribution; the latter remains fixed in most of our
discussion.

%%%%%%%%%%%%%%%%%%%%%%%%%%%%%%%%%%%%%%%%%%%%

\section{Problem Statement}\label{s:problem_statement}

Two parties observe correlated random variables $X$ and $Y$,
with \partyx observing $X$ and \partyy observing $Y$,
generated from a fixed distribution $\bPP {XY}$ and taking
values in finite sets $\cX$ and $\cY$, respectively. An {\it
  interactive protocol} $\prot$ (for these two parties)
consists of shared public randomness $U$, private
randomness\footnote{The random variables $U, U_\cX, U_\cY$
  are mutually independent and independent jointly of
  $(X,Y)$.}  $U_\cX$ and $U_\cY$, and interactive
communication $\Prot_1,\Prot_2, ..., \Prot_r$. 
The parties
communicate alternatively with \partyx transmitting in the
odd rounds and \partyy in the even rounds. Specifically,
in each round $i$ one of the party, say \partyx, communicates and transmits  
a string of bits $\Prot_i\in \{0,1\}^*$  determined by the previous
transmissions $\Prot_1, ..., \Prot_{i-1}$ and the observations $(X,
U_\cX, U)$ of the communicating party. To each possible value of the
bit string $\Pi_i$, a state from the state space $\{{\tt C}, \phi\}$
is associated. If the 
next state is ${\tt C}$, the other party starts communicating. If it is
$\phi$, the protocol stops and each party generates an output based on
its local observation and trascript $\Pi_1, ..., \Pi_i$ of the
protocol. We assume without loss of generality that \partyx initiates
the protocol. Note that the set $\cC_i$ of possible values of
$\Prot_i$, and the associated next states ${\tt C}$ or $\phi$ for each
value, is determined by a common function of $(X,U_\cX, U,
\Prot^{i-1})$ and $(Y,U_\cY, U, \Prot^{i-1})$ ($cf.$~\cite{GacKor73}),
\ie, as a function of a random variable $V$ such that
\[
H(V|X,U_\cX, U, \Prot^{i-1}) = H(V|Y,U_\cY, U, \Prot^{i-1}) =0.
\]
%A subclass often considered in literature restricts $\cC_i$ to be
%determined by $\Prot^{i-1}$. In this case, the  number of rounds of
%communication $r$ is a random 
%stopping-time such that the event $\{r = t\}$ is determined
%by the transcript $\Prot_1, ..., \Prot_t$. 
We denote the
overall transcript of the protocol
by $\Prot$. The {\it length of a protocol} $\prot$, $|\prot|$,
is the maximum number of bits that are communicated in any
execution of the protocol.

In the special case where $\cC_i$ is a prefix-free
set determined by $\Prot^{i-1}$, the protocol is called a {\it tree-protocol} ($cf.$ \cite{Yao79,
  KushilevitzNisan97}). In this case, the set of transcripts of the
protocol can be represented by a tree, termed the protocol tree, with
each leaf corresponding to a particular realization of the
transcript. Specifically, the protocol is defined by a binary tree
where each internal node $v$ is owned by either party, and node $v$ is 
labeled either by a function $a_v: \cX \times \cU_\cX \times \cU \to
\{0,1\}$ or $b_v:\cY \times \cU_\cY \times \cU \to \{0,1\}$. Then each
leaf, or the path from the root to the leaf, corresponds to the
overall transcript. Our proposed protocol is indeed a tree
protocol. On the other hand, our converse bound applies to the more
general class of interactive protocols described above. 

A random variable $F$ is said to be {\it recoverable} by
$\prot$ for \partyx (or \partyy) if $F$ is function of $(X,
U, U_\cX, \Prot)$ (or $(Y, U, U_\cY, \Prot)$).

A protocol with a constant $U$ is called a {\it private coin
  protocol}, with a constant $(U_\cX$, $U_\cY)$ is
called a {\it public coin protocol}, and with $(U, U_\cX,
U_\cY)$ constant is called a {\it deterministic protocol}. Note that a
private coin protocol can be realized as a public coin protocol by
sampling private coins from public coins. 

When we execute the protocol $\prot$ above, the overall {\it
  view} of the parties consists of random variables
$(XY\Pprot)$, where the two $\Prot$s correspond to the
transcript of the protocol seen by the two parties. A
simulation of the protocol consists of another protocol
which generates almost the same view as that of the original
protocol. We are interested in the simulation of private
coin protocols, using
arbitrary\footnote{\label{foot:no_private_randomness}Since
  we are not interested in minimizing the amount of
  shared randomness used in a simulation, we allow arbitrary public
  coin protocols to be used as simulation protocols.}
protocols; public coin protocols can be simulated by
simulating for each fixed value of public randomness the
resulting private coin protocol.

%%%%%
\begin{definition}[{\bf $\ep$-Simulation of a protocol}] 
\label{d:simulation}
Let $\prot$ be a private coin protocol. Given $0\le \ep <1$,
a protocol $\psim$ constitutes an $\ep$-simulation of
$\prot$ if there exist $\Protx$ and $\Proty$, respectively,
recoverable by $\psim$ for \partyx and \partyy such that
\begin{align}
\ttlvrn {\bPP{\Pprot XY}}{\bPP{\Protx\Proty XY}} \le \ep,
\label{e:simulation}
\end{align}
where $\ttlvrn {\dP} {\dQ} = \frac 12 \sum_x |\dP_x -
\dQ_x|$ denotes the variational or the statistical distance
between $\dP$ and $\dQ$.
\end{definition}
%%%%%
\begin{definition}[{\bf Distributional communication complexity}]
The $\ep$-error distributional communication complexity
$\dcpp$ of simulating a private coin protocol $\prot$ is the
minimum length of an $\ep$-simulation of $\prot$. The
distribution $\bPP{XY}$ remains fixed throughout our
analysis; for brevity, we shall abbreviate $\dcpp$ by
$\dcp$.
\end{definition}
%%%%%

\noindent {\bf Problem.} Given a protocol $\prot$ and a
joint distribution $\bPP{XY}$ for the observations of the
two parties, we seek to characterize $\dcp$.

\begin{remark}[{\bf Deterministic
      protocols}]\label{r:deterministic_protocols} Note that a
  deterministic protocol corresponds to 
an {\it interactive function}. A specific instance of this situation
appears in \cite{TVW15} where $\Prot(X,Y)=(X,Y)$ is considered.  
For such protocols,
\[
\ttlvrn {\bPP{\Pprot XY}}{\bPP{\Protx\Proty XY}} = 1 -
\bPr{\Prot = \Protx = \Proty}.
\] 
Therefore, a protocol is an $\ep$-simulation of a
deterministic protocol if and only if it computes the
corresponding interactive function with probability of error
less than $\ep$. Furthermore, randomization does not help in
this case, and it suffices to use deterministic simulation
protocols. Thus, our results below provide tight bounds for
distributional communication complexity of interactive
functions and even of all functions which are {\it
  information theoretically securely computable} for the
distribution $\bPP{XY}$, since computing these functions is
tantamount to computing an interactive function \cite{NTW15}
(see, also, \cite{Bea89, Kus92}).
\end{remark}

\begin{remark}[{\bf Compression of protocols}]\label{simulation_compression} 
A protocol $\pcom$ constitutes an $\ep$-compression of a
given protocol $\prot$ if it recovers $\Prot_\cX$ and
$\Prot_\cY$ for \partyx and \partyy such that
\[
\bPr{\Prot = \Prot_\cX = \Prot_\cY} \geq 1-\ep.
\]
Note that randomization does not help in this case either.
In fact, for deterministic protocols simulation and
compression coincide. In general, however, compression is a
more demanding task than simulation and our results show
that in many cases, (such as the amortized regime),
compression requires strictly more communication than
simulation.  Specifically, our results for $\ep$-simulation
in this paper can be modified to get corresponding results
for $\ep$-compression by replacing the information
complexity density $\icpp$ by
\[
h(\protr|x)+h(\protr|y) = -\log \bP{\Prot|X}{\protr|x}
\bP{\Prot|Y}{\protr|y}.
\]
The proofs remain essentially the same and, in fact,
simplify significantly.
\end{remark}

%%%%%%%%%%%%%%%%%%%%%%%%%%%%%%%%%%%%%%%%%%%%

\section{Main Results}\label{s:main_results}
We derive a lower bound for $\dcp$ which applies to all
private coin protocols $\prot$ and, in fact, applies to the
more general problem of communication complexity of sampling
a correlated random variable. For protocols with bounded
number of rounds of interaction, admittedly a significant restriction, 
\ie, protocols with $r =
r(X,Y,U, U_\cX, U_\cY) \leq \rmax$ with probability $1$, we
present a simulation protocol which yields upper bounds for
$\dcp$ of similar form as our lower bounds.  In particular,
in the asymptotic regime our bounds improve over previously
known bounds and are tight.

\subsection{Lower bound}\label{s:lower_bound}
We prove the following lower bound.
\begin{theorem}\label{t:lower_bound}
Given $0 \le \ep <1$ and a protocol $\prot$, for arbitrary
$0< \eta < 1/3$
\begin{align}
\dcp \geq \sup \{\la: \bPr{\icp > \la} \geq \ep +\epsmall\}
- \lasmall,
\label{e:lower_bound}
\end{align}
where the fudge parameters $\epsmall$ and $\lasmall$ depend
 on $\eta$ as
well as appropriately chosen information spectrums and will
be described below in \eqref{e:epsilon'} and
\eqref{e:lambda'}.
\end{theorem}
The appearance of fudge parameters such as $\epsmall$ and $\lasmall$
in the bound above is typical since 
the techniques to bound the tail probability of random variables 
invariably entail such parameters, which are tuned based on the specific 
scenario being studied. For instance, the Chernoff bound has a 
parameter that is tuned with respect to the moment generating 
function of the random variable of interest. More relevant to 
the problem studied here, such fudge parameters also show up in 
the evalutation of error probability of single-party non-interactive compression 
problems ($cf.$ \cite{HanVer93, Han03}). 

When the fudge parameters $\epsmall$ and $\lasmall$ are
negligible, the right-side of the bound above is close to
the $\ep$-tail of $\icp$. Indeed, the fudge parameters turn out
to be negligible in many cases of interest. For instance,
for the amortized case $\epsmall$ can be chosen to be
arbitrarily small. The parameter $\lasmall$ is related to
the length of the interval in which the underlying
information densities lie with probability greater than
$1-\epsmall$, the essential length of spectrums. For the
amortized case with product protocols, by the central limit
theorem the related essential spectrums are of length
$\Lambda = \Order(\sqrt{n})$ and $\lasmall = \log \Lambda$.
On the other hand, $\la_\ep$ is $\Order(n)$. Thus, the $\log
n$ order fudge parameter $\lasmall$ is negligible in this
case.  The same is true also for the example protocol in
Appendix~\ref{appendix:small-error}. Finally, it should be noted that
similar fudge parameters are ubiquitous in single-shot bounds;
for instance, see \cite[Lemma 1.3.2]{Han03}.

\begin{remark}\label{r:arbitrary_prot}
The result above does not rely on the interactive nature of
$\Prot$ and is valid for simulation of any random
variable $\Prot$. Specifically, for any joint distribution
$\bPP{\Prot XY}$, an $\ep$-simulation satisfying
\eqref{e:simulation} must communicate at least as many bits
as the right-side of \eqref{e:lower_bound}, which is roughly
equal to the largest value $\la_\ep$ of $\lambda$ such that
$\bPr{\icp > \la} > \ep$.
\end{remark}

{\bf The fudge parameters.} The fudge parameters $\epsmall$
and $\lasmall$ in Theorem~\ref{t:lower_bound} depend on the
spectrums of the following information densities:
\begin{enumerate}
\item[(i)] {\it Information complexity density:} This
  density is described in Definition~\ref{d:ICD} and will
  play a pivotal role in our results.

\item[(ii)] {\it Entropy density of $(X,Y)$:} This density,
  given by $ h(X,Y) = -\log\bP{XY}{X,Y}$, captures the
  randomness in the data and plays a fundamental role in the
  compression of the collective data of the two parties (\cf
  \cite{Han03}).

\item[(iii)] {\it Conditional entropy density of $X$ given
  $Y\Prot$:} The conditional entropy density $h(X|Y) = -\log
  \bP{X|Y}{X|Y}$ plays a fundamental role in the compression
  of $X$ for an observer of $Y$ \cite{MiyKan95, Han03}. We
  shall use the conditional entropy density $h(X|Y\Prot)$ in
  our bounds.

\item[(iv)] {\it Sum conditional entropy density of
  $(X\Prot, Y\Prot)$:} The sum conditional entropy density
  is given by $\hsum {X}{Y} = -\log
  \bP{X|Y}{X|Y}\bP{Y|X}{Y|X}$ has been shown recently to
  play a fundamental role in the communication complexity of
  the data exchange problem \cite{TVW15}. We shall use the
  sum conditional entropy density $\hsump$.

\item[(v)] Information density of $X$ and $Y$ is given by
  $i(X\wedge Y) \ed h(X) - h(X|Y)$.
\end{enumerate}
Let $[\lamins 1, \lamaxs 1]$, $[\lamins 2, \lamaxs 2]$, and
$[\lamins 3, \lamaxs 3]$ denote the ``essential'' spectrums
of information densities $\den_1 = \hxy$, $\den_2 =
h(X|Y\Prot)$, and $\den_3=\hsump$, respectively. Concretely,
let the tail events $\cE_i = \{\den_i \notin [\lamins i,
  \lamaxs i]\}$, $i =1,2,3$, satisfy
\begin{align}
\bPr{ \cE_1}+\bPr{\cE_2}+\bPr{\cE_3} \leq \eptail,
\label{e:tail_prob}
\end{align}
where $\eptail$ can be chosen to be appropriately
small. Further, let $\Lambda_i = \lamaxs i - \lamins i$, $i
=1,2,3$, denote the corresponding effective spectrum
lengths. The parameters $\epsmall$ and $\lasmall$ in
Theorem~\ref{t:lower_bound} are given by
\begin{align}
\epsmall = \eptail + 2\eta
\label{e:epsilon'}
\end{align}
and
\begin{align}
\lasmall = 2\log \Lambda_1 \Lambda_3 +\log\Lambda_2
-\log(1-3\eta) + 9\log 1/\eta + 3,
\label{e:lambda'}
\end{align}
where $0<\eta<1/3$ is arbitrary. If $\Lambda_i = 0$,
$i=1,2,3$, we can replace it with $1$ in the bound
above. Thus, our spectrum slicing approach allows us to
reduce the dependence of $\lasmall$ on spectrum lengths
$\Lambda_i$'s from linear to logarithmic.

%%%%%%%%%%%%%%%%%%%%%%%%%%%%%%%%%%%%%%%%%%%%%%%%%%%%

\subsection{Upper bound}
We prove the following upper bound.
\begin{theorem} \label{t:upper_bound}
For every $0 \le \ep < 1$ and every protocol $\prot$,
\begin{align}
\dcp \le \inf\left\{ \lambda : \bPr{ \icp> \la} \le \ep -
\epsmall \right\} + \lasmall, \nonumber
\end{align}
where the fudge parameters $\epsmall$ and $\lasmall$ depend
on the maximum number of rounds of interaction in $\prot$
and on appropriately chosen information spectrums.
\end{theorem}

\begin{remark}
In contrast to the lower bound given in the previous
section, the upper bound above relies on the interactive
nature of $\prot$. Furthermore, the fudge parameters
$\epsmall$ and $\lasmall$ depend on the number of rounds,
and the upper bound may not be useful when the number of
rounds is not negligible compared to the $\ep$-tail of the
information complexity density. However, we will see that
the above upper bound is tight for the amortized regime,
even up to the second-order asymptotic term.
\end{remark}

{\bf The simulation protocol.}  Our simulation protocol
simulates the given protocol $\prot$ round-by-round,
starting from $\Prot_1$ to $\Prot_r$. Simulation of each
round consists of two subroutines: Interactive Slepian-Wolf
compression and message reduction by public randomness.

The first subroutine uses an interactive version of the
classical Slepian-Wolf compression \cite{SleWol73} (see
\cite{MiyKan95} for a single-shot version) for sending $X$
to an observer of $Y$. The standard (noninteractive)
Slepian-Wolf coding entails hashing $X$ to $l$ values and
sending the hash values to the observer of $Y$.  The number
of hash values $l$ is chosen to take into account the
worst-case performance of the protocol. However, we are not
interested in the worst-case performance of each round, but
of the overall multiround protocol. As such, we seek to
compress $X$ using the least possible instantaneous rate. To
that end, we increase the number of hash values gradually,
$\Delta$ at a time, until the receiver decodes $X$ and sends
back an ACK. We apply this subroutine to each round $i$, say
$i$ odd, with $\Prot_i$ in the role of $X$ and $(Y,
\Prot_1...., \Prot_{i-1})$ in the role of $Y$. Similar
interactive Slepian-Wolf compression schemes have been
considered earlier in different contexts (\cf~\cite{FedS02,
  Orlitsky90, YanH10, HayTyaWat14ii, TVW15}).

The second subroutine reduces the number of bits
communicated in the first by realizing a portion of the
required communication by the shared public randomness
$U$. Specifically, instead of transmitting hash values of
$\Prot_i$, we transmit hash values of a random variable
$\hat{\Prot}_i$ generated in such a manner that some of its
corresponding hash bits can be extracted from $U$ and the
overall joint distributions do not change by much. Since $U$
is independent of $(X,Y)$, the number $k$ of hash bits that
can be realized using public randomness is the maximum
number of random hash bits of $\Prot_i$ that can be made
almost independent of $(X,Y)$, a good bound for which is
given by the {leftover hash lemma}. The overall simulation
protocol for $\Prot_i$ now communicates $l-k$ instead of $l$
bits.  A similar technique for message reduction appears in
a different context in \cite{RenRen11, Mur14, YasAreGoh14}.

The overall performance of the protocol above is still
suboptimal because the saving of $k$ bits is limited by the
worst-case performance. To remedy this shortcoming, we once
again take recourse to spectrum slicing to ensure that our
saving $k$ is close to the best possible for each
realization $(\Prot, X,Y)$.

Note that our protocol above is closely related to that
proposed in \cite{BraRao11}. However, the information
theoretic form here makes it amenable to techniques such as
spectrum slicing, which leads to tighter bounds than those
established in \cite{BraRao11}.

{\bf The fudge parameters.} The fudge parameters $\epsmall$
and $\lasmall$ in Theorem~\ref{t:upper_bound} depend on the
spectrum of various conditional information densities. To optimize the
performance of each subroutine described above, we slice the spectrum of the
respective conditional information density involved. Specifically, for
odd round $t$, we slice the spectrum of $h(\Prot_t|Y
\Prot^{t-1}) = - \log
\bP{\Prot_t|Y\Prot^{t-1}}{\Prot_t|Y,\Prot^{t-1}}$ for
interactive Slepian-Wolf compression and $h(\Prot_t|X
\Prot^{t-1}) = - \log
\bP{\Prot_t|X\Prot^{t-1}}{\Prot_t|X,\Prot^{t-1}}$ for the
substitution of message by public randomness; for even
rounds, the role of $X$ and $Y$ is interchanged.  Each round
involves some residuals related to the two conditional
information densities.  The fudge parameters
$\epsmall$ and $\lasmall$ are accumulations of the residuals
of each round. The explicit expressions for $\epsmall$ and $\lasmall$
are rather technical and are given in Section~\ref{s:simulation_final}
along with the proofs.

%%%%%%%%%%%%%%%%%%%%%%%%%%%%%%%%%%%%%%%%%%%%%%%%%%%%%%%%

%%%%%%%%%%%%%%%%%%%%%%%%%%%%%%%%%%%%%%%%%%%%%%%%%%%%%%%%

\subsection{Amortized regime: second-order asymptotics}
It was shown in \cite{BraRao11} that information complexity
of a protocol equals the amortized communication rate for
simulating the protocol, \ie,
\begin{align} \nonumber
\lim_{\ep \to 0} \lim_{n \to \infty} \frac{1}{n}
D_\ep(\prot^n|\bPP{XY}^n) = \Icp,
\end{align}
where $\bPP{XY}^n$ denotes the $n$-fold product of the
distribution $\bPP{XY}$, namely the distribution of random
variables $(X_i, Y_i)_{i=1}^n$ drawn IID from $\bPP{XY}$,
and $\prot^n$ corresponds to running the same protocol
$\prot$ on every coordinate $(X_i, Y_i)$.  Thus, $\Icp$ is
the first-order term (coefficient of $n$) in the
communication complexity of simulating the $n$-fold product
of the protocol.  However, the analysis in \cite{BraRao11}
sheds no light on finer asymptotics such as the second-order
term or the dependence of $D_\ep(\pi^n|\bPP{XY}^n)$
on\footnote{The lower bound in \cite{BraRao11} gives only
  the {\it weak converse} which holds only when $\ep = \ep_n
  \rightarrow0$ as $n\rightarrow\infty$.} $\ep$. On the one
hand, it even remains unclear from \cite{BraRao11} if a
positive $\ep$ reduces the amortized communication rate or
not.  On the other hand, the amortized communication rate
yields only a loose bound for $D_\ep(\prot^n| \bPP{XY}^n)$
for a finite, fixed $n$.  A better estimate of
$D_\ep(\prot^n| \bPP{XY}^n)$ at a finite $n$ and for a fixed
$\ep$ can be obtained by identifying the second-order
asymptotic term. Such second-order asymptotics were first
considered in \cite{Str62} and have received a lot of
attention in information theory in recent years following
\cite{Hay09, PolPooVer10}.

Our lower bound in Theorem~\ref{t:lower_bound} and upper
bound in Theorem~\ref{t:upper_bound} show that the leading
term in $D_\ep(\prot^n| \bPP{XY}^n)$ is roughly the
$\ep$-tail $\la_\ep$ of the random variable $\ic(\Prot^n;
X^n, Y^n) = \sum_{i=1}^n \ic(\Prot_i; X_i, Y_i)$, a sum of
$n$ IID random variables. By the central limit theorem the
first-order asymptotic term in $\la_\ep$ equals
$n\bEE{\ic(\Prot; X, Y)} = n\Icp$, recovering the result of
\cite{BraRao11}. Furthermore, the second-order asymptotic
term depends on the variance $\varp$ of $\icp$, \ie, on
\begin{align*}
\varp \ed \mathrm{Var}\left[ \icp \right].
\end{align*}
We have the following result.
\begin{theorem}\label{theorem:second-order}
For every $0 < \ep < 1$ and every protocol $\prot$ with
$\mathtt{V}(\prot) > 0$,
\begin{align*}
D_\ep(\prot^n|\bPP{XY}^n) = n \Icp + \sqrt{n \varp}
Q^{-1}(\ep) + \order(\sqrt{n}),
\end{align*}
where $Q(x)$ is equal to the probability that a standard
normal random variable exceeds $x$.
\end{theorem}
As a corollary, we obtain the {\it strong
  converse}.
\begin{corollary} \label{corollary:strong-converse}
For every $0<\ep<1$, the amortized communication rate
\[
\lim_{n \to \infty} \frac{1}{n} D_\ep(\prot^n|\bPP{XY}^n) =
\Icp.
\]
\end{corollary}

%%%%%%%%%%%%%%%%%%%%%%%%%%%%%%%%%%%%%%%%%%%%%%%%%%%%%%%
%\subsection{Direct product theorem for simulation of protocols}
Corollary \ref{corollary:strong-converse} implies that the
amortized communication complexity of simulating protocol
$\prot$ cannot be smaller than its information complexity
even if we allow a positive error.  Thus, if the length of
the simulation protocol $\psim$ is ``much smaller'' than $n
\Icp$, the corresponding simulation error $\ep = \ep_n$ must
approach $1$.  But how fast does this $\ep_n$ converge to
$1$?  Our next result shows that this convergence is
exponentially rapid in $n$.
%%%%
\begin{theorem} \label{theorem:direct-product}
Given a protocol $\prot$ and an arbitrary $\delta > 0$, for
any simulation protocol $\psim$ with
\begin{align*}
 | \psim | \le n[\Icp - \delta],
\end{align*}
there exists a constant $E = E(\delta)>0$ such that for
every $n$ sufficiently large, it holds that
%an integer $n_0(\delta)$
%and a constant $E = E(\delta) > 0$ such that, for every $n \ge n_0(\delta)$, protocol $\psim$
%simulating $\prot^n$ with 
%must satisfy
\begin{align*}
\ttlvrn {\bPP{\Prot^n \Prot^n X^nY^n}}{\bPP{\Protx^n
    \Proty^n X^nY^n}} \ge 1- 2^{- E n}.
\end{align*}
\end{theorem}
A similar converse was first shown for the channel coding
problem by Arimoto \cite{Ari73} (see
\cite{DueckK79, PolVer10} for further refinements of this
result), and has been studied for other classical
information theory problems as well. To the best of our
knowledge, Corollary~\ref{theorem:direct-product} is the
first instance of an Arimoto converse for a problem
involving interactive communication.

In the theoretical computer science literature, such converse results have been
termed {\it direct product theorems} and have been
considered in the context of the (distributional)
communication complexity problem (for computing a given
function) \cite{BraRaoWeiYeh13, BraWei15, JaiPerYao12}.  Our
lower bound in Theorem~\ref{t:lower_bound}, too, yields a
direct product theorem for the communication complexity
problem.  We state this simple result in the passing,
skipping the details since they closely mimic
Theorem~\ref{theorem:direct-product}.  Specifically, given a
function $f$ on $\cX \times \cY$, by a slight abuse of
notations and terminologies, let $D_\ep(f) =
D_\ep(f|\bPP{XY})$ be the communication complexity of
computing $f$. As noted in Remark~\ref{r:arbitrary_prot},
Theorem~\ref{t:lower_bound} is valid for an arbitrary random
variables $\Prot$, and not just an interactive
protocol. Then, by following the proof of
Theorem~\ref{theorem:direct-product} with $F = f(X,Y)$
replacing $\Prot$ in the application of
Theorem~\ref{t:lower_bound}, we get the following direct
product theorem.
\begin{theorem} \label{theorem:direct-product-function}
Given a function $f$ and an arbitrary $\delta > 0$, for any
function computation protocol $\prot$ computing estimates
$F_{\cX,n}$ and $F_{\cY,n}$ of $f^n$ at the \partyx and
\partyy, respectively, and with length
\begin{align}
 | \prot | \le n[H(F|X) + H(F|Y) - \delta],
\label{e:threshold_f}
\end{align}
there exists a constant $E = E(\delta)>0$ such that for
every $n$ sufficiently large, it holds that
%an integer $n_0(\delta)$
%and a constant $E = E(\delta) > 0$ such that, for every $n \ge n_0(\delta)$, protocol $\psim$
%simulating $\prot^n$ with 
%must satisfy
\begin{align*}
\bPr{F_{\cX,n} = F_{\cY,n} = F^n} \le 2^{- E n},
\end{align*}
where $F^n := (F_1, ..., F_n)$ and $F_i := f(X_i, Y_i)$,
$1\le i \le n$.
\end{theorem}
Recall that \cite{BraRao11, MaIsh11} showed that the first
order asymptotic term in the amortized communication
complexity for function computation equals the
information complexity $\IC(f)$ of the function, namely the
infimum over $\Icp$ for all interactive protocols $\prot$
that recover $f$ with $0$ error. Ideally, we would like to
show an Arimoto converse for this problem, \ie, replace the
threshold on the right-side of \eqref{e:threshold_f} with
$n[\IC(f) -\delta]$. The direct product result above is
weaker than such an Arimoto converse, and proving the
Arimoto converse for the function computation problem is
work in progress. Nevertheless, the simple result above is
not comparable with the known direct product theorems in
\cite{BraRaoWeiYeh13, BraWei15} and can be stronger in some
regimes\footnote{The result in \cite{BraRaoWeiYeh13,
    BraWei15} shows a direct product theorem when we
  communicate less than $n\IC(f)/\mathtt{poly}(\log n)$.}.
%was shown to be 
%Although such a direct product theorem is loose when
%$D_\ep(f) \gg \mathtt{IC}(f)$, it is tighter than previously known direct product theorems 
%when $D_\ep(f) \simeq \mathtt{IC}(f)$.

%%%%%%%%%%%%%%%%%%%%%%%%%%%%%%%%%%%%%%%%%%%%%%%%%%%%%%%%
%%%%%%%%%%%%%%%%%%%%%%%%%%%%%%%%%%%%%%%%%%%%%%%%%%%%%%%%

\subsection{General formula for amortized communication complexity}
Consider arbitrary distributions $\bPP{X_nY_n}$ on
$\cX^n\times \cY^n$ and arbitrary protocols $\prot_n$ with
inputs $X_n$ and $Y_n$ taking values in $\cX^n$ and $\cY^n$,
for each $n\in \mN$.  For vanishing simulation error
$\ep_n$, how does $D_{\ep_n}(\prot_n| \bPP{X_n Y_n})$ evolve
as a function of $n$?

The previous section, and much of the theoretical computer
science literature, has focused on the case when $\bPP
{X_nY_n} = \bPP{XY}^n$ and the same protocol $\prot$ is
executed on each coordinate.  In this section, we identify
the first-order asymptotic term in $D_{\ep_n}(\prot_n|
\bPP{X_nY_n})$ for a general sequence of
distributions\footnote{We do not require $\bPP{X_n Y_n}$ to
  be even consistent.}  $\{ \bPP {X_nY_n} \}_{n=1}^\infty$
and a general sequence of protocols $\prots = \{
\prot_n\}_{n=1}^\infty$. Formally, the amortized
(distributional) communication complexity of $\prots$ for
$\{ \bPP {X_nY_n} \}_{n=1}^\infty$ is given
by\footnote{Although $D(\prots)$ also depends on $\{ \bPP
  {X_nY_n} \}_{n=1}^\infty$, we omit the dependency in our
  notation.}
\begin{align*}
D(\prots)\ed \lim_{\ep \to 0} \,\limsup_{n\to\infty}
\frac{1}{n} D_\ep(\prot_n | \bPP{X_nY_n}).
\end{align*}

Our goal is to characterize $D(\prots)$ for any given
sequences $\bPP n$ and $\prots$.  We seek a general formula
for $D(\prots)$ under minimal assumptions.  Since we do not
make any assumptions on the underlying distribution, we
cannot use any measure concentration results. Instead, we
take recourse to probability limits of information spectrums
introduced by Han and Verd\'u in \cite{HanVer93} for
handling this situation (\cf~\cite{Han03}).  Specifically,
for a sequence of protocols $\boldsymbol{\prot} = \{ \prot_n
\}_{n=1}^\infty$ and a sequence of observations
$(\mathbf{X},\mathbf{Y}) = \{ (X_n,Y_n) \}_{n=1}^\infty$,
the {\em sup information complexity} is defined as
\begin{align*}
\sIcp \ed \inf\left\{ \alpha \mid \lim_{n\to\infty} \bPr{
  \frac{1}{n} \mathtt{ic}(\Prot_n;X_n,Y_n) > \alpha } = 0
\right\},
\end{align*} 
where, with a slight abuse of notation, $\Prot_n$ is the
transcript of protocol $\prot_n$ for observations
$(X_n,Y_n)$.
%\begin{align}
%\sIcp = \Icp
%\label{e:sIC_IC}
%\end{align}
The result below shows that it is $n\sIcp$, and not
$\IC(\prot_n)$, that determines the communication complexity
in general.
\begin{theorem} \label{theorem:general}
For every sequence of protocols $\boldsymbol{\prot} = \{
\prot_n \}_{n=1}^\infty$,
\begin{align*}
D(\prots) = \overline{\mathtt{IC}}(\boldsymbol{\prot}).
\end{align*}
\end{theorem}
\noindent The proof uses Theorem~\ref{t:lower_bound} and
Theorem~\ref{t:upper_bound} with carefully chosen
spectrum-slice sizes.

For the case when $\prot_n = \prot^n$ and $\bPP{X_n Y_n} =
\bPP{XY}^n$, it follows from the law of large numbers that
$\sIcp = \Icp$ and we recover the result of
\cite{BraRao11}. However, the utility of the general formula
goes beyond this simple amortized regime.
Example~\ref{ex:mixed_protocol} provides one such
instance. In this case, $\sIcp$ can be easily shown to equal
$\IC(\prot_\mathtt{h})$ for any bias of the coin $\Prot_0$.

%%%%%%%%%%%%%%%%%%%%%%%%%%%%%%%%%%%%%%%%%%%%

\section{Background: Secret Key Agreement and Data Exchange}\label{s:background}

Our proofs draw from various techniques in cryptography and
information theory.  In particular, we use our recent
results on information theoretic secret key agreement and
data exchange, which are reviewed in this section together
with the requisite background.

%%%%%%%%%%%%%%%%%%%%%%%%%%%%%%%%%%
\subsection{Secret key agreement by public discussion}
The problem of two party secret key agreement by public
discussion was alluded to in \cite{BenBraRob88}, but a
proper formulation and an asymptotically optimal
construction appeared first in \cite{Mau93,
  AhlCsi93}. Consider two parties with the first and the
second party, respectively, observing the random variable
$X$ and $Y$.  Using an interactive protocol $\prot$ and
their local observations, the parties agree on a secret
key. A random variable $K$ constitutes a secret key if the
two parties form estimates that agree with $K$ with
probability close to $1$ and $K$ is concealed, in effect,
from an eavesdropper with access to the transcript $\Prot$
and a side-information $Z$.  Formally, let $K_\cX$ and
$K_\cY$, respectively, be recoverable by $\prot$ for the
first and the second party. Such random variables $K_\cX$
and $K_\cY$ with common range $\cK$ constitute an
{\it$\ep$-secret key} if the following condition is
satisfied:
\begin{eqnarray}
\ttlvrn{\bPP{K_\cX K_\cY\Prot Z}}{
  \mathrm{P}_{\mathtt{unif}}^{(2)}\times \bPP{\Prot Z}}
&\leq \ep, \nonumber
\end{eqnarray}
where
\begin{eqnarray*}
\mathrm{P}_{\mathtt{unif}}^{(2)}\left(k_\cX, k_\cY\right) =
\frac{\mathbbm{1}(k_\cX= k_\cY)}{|\cK|}.
\end{eqnarray*}
The condition above ensures both reliable {\it recovery},
requiring $\bPr{K_\cX \neq K_\cY}$ to be small, and
information theoretic {\it secrecy}, requiring the
distribution of $K_\cX$ (or $K_\cY$) to be almost
independent of the eavesdropper's side information
$(\Prot,Z)$ and to be almost uniform. See \cite{TyaWat14}
for a discussion.

\begin{definition}
Given $0\le \ep < 1$, the supremum over lengths $\log|\cK|$
of an $\ep$-secret key is denoted by $S_\ep(X, Y|Z)$, and
for the case when $Z$ is constant by $S_\ep(X,Y)$.
\end{definition}

By its definition, $S_\ep(X,Y|Z)$ has the following
monotonicity property.
\begin{lemma}[{\bf Monotonicity}]\label{l:monotonicity}
For any deterministic protocol $\prot$,
\[
S_\ep(X,Y|Z) \geq S_\ep(X\Prot, Y\Prot| Z \Prot).
\]
Furthermore, if $V_\cX$ and $V_\cY$ can be recovered by
$\prot$ for the first and the second party, respectively,
then
\[
S_\ep(X,Y|Z) \geq S_\ep(XV_\cX, V_\cY| Z \Prot).
\]
\end{lemma}
\noindent The claim holds since the two parties can generate
a secret key by first running $\prot$ and then generating a
secret key for the case when the first party observes $(X,
\Prot)$, the second party observes $(Y, \Prot)$ and the
eavesdropper observes $(Z, \Prot)$. Similarly, the second
inequality holds since the parties can ignore a portion of
their observations and generate a secret key from
$(X,V_\cX)$ and $(Y, V_\cY)$.

%%%%%%%%%%%%%%%%%%%%
\subsubsection{Leftover hash lemma}
A key tool for generating secret keys is the {\it leftover
  hash lemma} ($cf.$~\cite{Ren05}) which, given a random
variable $X$ and an $l$-bit eavesdropper's observation $Z$,
allows us to extract roughly $H_{\min}(\bPP X) - l$ bits of
uniform bits, independent of $Z$.  We shall use a slightly
more general form. Given random variables $X$ and $Z$, let
\[
\Hmin\lc\bPP{XZ}\mid \bQQ{Z}\rc \ed \sup_{x,z} - \log
\frac{\bP {XZ}{x,z}}{\bQ Z z}.
\]
We define the {\it conditional min-entropy}
of $X$ given $Z$ as
\[
\Hmin\lc \bPP{XZ}\mid Z\rc \ed \sup_{\bQQ Z \,:\, \supp(\bPP
  Z)\, \subset\, \supp(\bQQ Z)} \Hmin\lc \bPP {XZ}\mid
\bQQ{Z}\rc.
\]
Further, let $\cF$ be a {\it 2-universal family} of mappings
$f: \cX\rightarrow \cK$, $i.e.$, for each $x'\neq x$, the
family $\cF$ satisfies
\[
\frac{1}{|\cF|} \sum_{f\in \cF} \mathbbm{1}(f(x) = f(x'))
\leq \frac{1}{|\cK|}.
\]

\begin{lemma}[{\bf Leftover Hash}]\label{l:leftover_hash} Consider random variables $X, Z$ and $V$ taking values in
countable sets $\cX$, $\cZ$, and a finite set $\cV$,
respectively.  Let $S$ be a random seed such that $f_S$ is
uniformly distributed over a 2-universal family $\cF$.
Then, for $K_S = f_S(X)$
\begin{align*}
\bE {S}{\ttlvrn{\bPP{K_SVZ}}{\bPP{\mathtt{unif}}\bPP{VZ}}}
\leq \frac 12\sqrt{|\cK||\cV| 2^{- \Hmin \lc \bPP{XZ}\mid
    Z\rc}},
\end{align*}
where $\bPP{\mathtt{unif}}$ is the uniform distribution on
$\cK$.
\end{lemma}
The version above is a straightforward modification of the
leftover hash lemma in, for instance, \cite{Ren05} and can
be derived in a similar manner.

As an application of the leftover hash lemma above, we get
the following useful result.
\begin{lemma}\label{l:SK_relation}
Consider random variables $X, Y, Z$ and $V$ taking values in
countable sets $\cX$, $\cY$, $\cZ$, and a finite set $\cV$,
respectively.  Then,
\[
S_{2\ep}(X,Y|ZV) \geq S_\ep(X,Y|Z) - \log|\cV| -
2\log(1/2\ep).
\]
\end{lemma}
The proof is relegated to Appendix~\ref{a:proof_lemma}.

%%%%%
\subsubsection{Conditional independence testing upper bound for secret key lengths}

Next, we recall the {\it conditional independence testing}
upper bound for $S_{\ep}(X, Y)$, which was established in
\cite{TyaWat14, TyaWat14ii}. In fact, the general upper
bound in \cite{TyaWat14, TyaWat14ii} is a single-shot upper
bound on the secret key length for a multiparty secret key
agreement problem with side information at the
eavesdropper. Below, we recall a specialization of the
general result for the two party case with no side
information at the eavesdropper. In order to state the
result, we need the following concept from binary hypothesis
testing.

Consider a binary hypothesis testing problem with null
hypothesis $\mathrm{P}$ and alternative hypothesis
$\mathrm{Q}$, where $\mathrm{P}$ and $\mathrm{Q}$ are
distributions on the same alphabet ${\cal V}$. Upon
observing a value $v\in \cV$, the observer needs to decide
if the value was generated by the distribution $\bPP{}$ or
the distribution $\mathrm{Q}$. To this end, the observer
applies a stochastic test $\mathrm{T}$, which is a
conditional distribution on $\{0,1\}$ given an observation
$v\in \cV$. When $v\in \cV$ is observed, the test
$\mathrm{T}$ chooses the null hypothesis with probability
$\mathrm{T}(0|v)$ and the alternative hypothesis with
probability $T(1|v) = 1 - T(0|v)$.  For $0\leq \ep<1$,
denote by $\beta_\ep(\mathrm{P},\mathrm{Q})$ the infimum of
the probability of error of type II given that the
probability of error of type I is less than $\ep$, \ie,
\begin{eqnarray}
\beta_\ep(\mathrm{P},\mathrm{Q}) := \inf_{\mathrm{T}\, :\,
  \mathrm{P}[\mathrm{T}] \ge 1 - \ep}
\mathrm{Q}[\mathrm{T}],
%\label{e:beta-epsilon}
\nonumber
\end{eqnarray}
where
\begin{eqnarray*}
\mathrm{P}[\mathrm{T}] &=& \sum_v \mathrm{P}(v)
\mathrm{T}(0|v), \\ \mathrm{Q}[\mathrm{T}] &=& \sum_v
\mathrm{Q}(v) \mathrm{T}(0|v).
\end{eqnarray*}

The following upper bound for $S_\ep(X,Y)$ was established
in \cite{TyaWat14, TyaWat14ii}.
\begin{theorem}[{\bf Conditional independence testing bound}] 
\label{theorem:one-shot-converse-source-model}
Given $0\leq \ep <1$, $0<\eta<1-\ep$, the following bound
holds:
\begin{eqnarray}
S_{\ep}\left(X, Y\right) \le -\log
\beta_{\ep+\eta}\big(\bPP{XY},\mathrm{Q}_{X}\mathrm{Q}_{Y}\big)
+ 2 \log(1/\eta), \nonumber
\end{eqnarray}
for all distributions $\bQQ X$ and $\bQQ Y$ on $\cX$ and
$\cY$, respectively.
\end{theorem}
We close by noting a further upper bound for
$\beta_\ep(\bPP{}, \bQQ{})$, which is easy to derive.
\begin{lemma}\label{l:bound_beta_epsilon}
For every $0\leq \ep < 1$ and $\lambda$,
\[
-\log \beta_\ep(\bPP{}, \bQQ{}) \leq \lambda -
\log\left(\mathrm{P}\left(\log\frac{ \bP {} X}{\bQ{} X} <
\lambda\right) - \ep\right)_+,
\]
where $(x)_+ = \max\{0,x\}$. As a corollary, we obtain the
following upper bound for $S_\ep(X,Y)$:
\[
S_{\ep}\left(X, Y\right) \le \lambda - \log\lc \bPr{\log
  \frac{\bP{XY}{X,Y}}{\bQ X X \bQ Y Y} < \lambda} - \ep -
\eta\rc_+ + 2 \log(1/\eta),
\]
for all distributions $\bQQ X$ and $\bQQ Y$.
\end{lemma}

%%%%%%%%%%%%%%%%%%%%
\subsection{The data exchange problem} \label{section:data-exchange}

The next primitive that will be used in the reduction
argument in our lower bound proof is a protocol for data
exchange. The parties observing $X$ and $Y$ seek to know
each other's data. What is the minimum length of interactive
communication required? This basic problem, first studied in
\cite{OrlEl84}, is in effect a two-party extension of the
classical Slepian-Wolf compression \cite{SleWol73} (see
\cite{CsiNar04} for a multiparty version). In a recent work
\cite{TVW15}, we derived tight lower and upper bounds for
the length of a protocol that, for a given distribution
$\bPP{XY}$, will facilitate data exchange with probability
of error less than $\ep$.  We review the proposed
protocol and its performance here; first, we formally define
the data exchange problem.

%%%%%
\begin{definition}
For $0 \le \ep < 1$, a protocol $\prot$ attains $\ep$-{\em
  data exchange} if there exist $\hat{Y}$ and $\hat{X}$
which are recoverable by $\prot$ for the first and the
second party, respectively, and satisfy
\begin{align*}
\bPP{}(\hat{X} = X,~\hat{Y} = Y) \ge 1 - \ep.
\end{align*}
\end{definition}
%%%%%

Note that data exchange corresponds to simulating a
(deterministic) interactive protocol $\prot$ where
$\Prot_1(X) = X$ and $\Prot_2 = Y$; attaining $\ep$-data
exchange is tantamount to $\ep$-simulation of $\prot$. In
fact, the specific protocol for data exchange proposed in
\cite{TVW15} can be recovered as a special case of our
simulation protocol in
Section~\ref{s:simulation_protocols}. The next result
paraphrases {\cite[Theorem 2]{TVW15}} and can also be
recovered as a special case of
Lemma~\ref{l:simulation_general}.

We paraphrase the result form \cite{TVW15} in a form that is
more suited for our application here. The data exchange
protocol proposed in \cite{TVW15} relies on slicing the
spectrum of $h(X|Y)$ (or $h(Y|X)$). Let
$\cE_{\mathtt{tail}}$ denote the tail event $h(X|Y) \notin
[\lamin',\lamax']$. The protocol entails slicing the
essential spectrum $[\lamin', \lamax']$ into $N$ parts of
length $\Delta$ each, \ie,
\[
N = \frac{\lamax' -\lamin'}{\Delta}.
\]
\begin{theorem}[{\cite[Theorem 2]{TVW15}, Lemma~\ref{l:simulation_general}}]\label{t:DE_protocol}
Given $\Delta>0, \xi >0$, and $N$ as above, there exists a
deterministic protocol for $\ep$-data exchange satisfying
the following properties:
\begin{itemize}
\item[(i)] Denoting by $\cE_{\mathtt{error}}$ the error
  event, it holds that
\[
\bP{XY}{\cE_{\mathtt{error}}\cap \{h(X\triangle Y) \leq
  \lambda\}} \le \bP {XY}{\cE_{\mathtt{tail}}} + N2^{-\xi},
\]
which further yields that the probability of error $\ep$ is
bounded above as
\[
\ep \le \bP{XY}{h(X\triangle Y) > \lambda} + \bP
    {XY}{\cE_{\mathtt{tail}}} + N2^{-\xi};
\]

\item[(ii)] the protocol communicates no more than $\lambda
  + \Delta + N+\xi$ bits;

\item[(iii)] for every $(X, Y)$ such that $\lamin' < h(X|Y)
  < \lamax'$, the transcript of the protocol can take no
  more than $2^{\hsum X Y + \Delta + \xi}$ values.
\end{itemize}
\end{theorem}
Note that property (iii) above, though not explicitly stated
in \cite[Theorem 2]{TVW15} or in the general
Lemma~\ref{l:simulation_general} below, follows simply from
the proofs of these results. It makes the subtle observation
that while, for each $(X, Y)$ such that $\lamin' < h(X|Y) <
\lamax'$, $\hsum X Y + \Delta + N+ \xi$ bits are
communicated to interactively generate the transcript, the
number of (variable length) transcripts is no more
than\footnote{The $N$-bit ACK-NACK feedback used in the
  protocol can be determined from the length of the
  transcript.} $\hsum X Y + \Delta + N+ \xi$.  Property (ii)
above was crucial to establish the communication complexity
results of \cite{TVW15}; property (iii) was not relevant in
the context of that work. On the other hand, here we shall
use the protocol of Theorem~\ref{t:DE_protocol} in our
reduction to secret key agreement in the next section and
will treat the communication used in data exchange as
eavesdropper's side information. As such, it suffices to
bound the number of values taken by the transcript; the
number of bits actually communicated in the interactive
protocol is a loose upper bound on the former quantity.

Interestingly, our simulation protocol given
in Section~\ref{s:simulation_protocols} is used both in our
upper bound to compress a given protocol and in our lower
bound to complete the reduction argument.

%%%%%%%%%%%%%%%%%%%%%%%%%%%%%%%%%%%%%%%%%%%%

\section{Proof of Lower Bound}\label{s:lower_bound_proof}
As described in the introduction, our proof of
Theorem~\ref{t:lower_bound} relies on generating a secret
key for $X$ and $Y$ from a given $\ep$-simulation $\psim$ of
$\prot$.  However, there are two caveats in the heuristic
approach described in the introduction:

\noindent {First}, to extract secret keys from the generated
common randomness we rely on the {leftover hash lemma}. In
particular, the bits are extracted by applying a 2-universal
hash family to the common randomness generated. However, the
range-size of the hash family must be selected based on the
min-entropy of the generated common randomness, which is not
easy to estimate. To remedy this, we communicate more using
a data-exchange protocol proposed in \cite{TVW15} to make
the collective observations $(X,Y)$ available to both the
parties; a good bound for the communication complexity of
this protocol is available. The generated common randomness
now includes $(X,Y)$ for which the min-entropy can be easily
bounded and the size of the aforementioned extracted secret
key can be tracked. A similar {\it common randomness
  completion and decomposition} technique was introduced in
\cite{TyaThesis} to characterize a class of securely
computable functions.

\noindent {Second}, our methodology described above requires
bounds on various information densities in different
directions. A direct application of this method will result
in a gap equal to the effective length of various spectrums
involved. To remedy this, we apply the methodology described
above not to the original distribution $\bPP{XY}$ but a
conditional distribution $\bPP{XY|\cE}$ where the event
$\cE$ is an appropriately chosen event contained in single
slices of various spectrums involved. Such a conditioning is
allowed since we are interested in the worst-case
communication complexity of the simulation protocol.

We now describe the proof of Theorem~\ref{t:lower_bound} in
detail.  To make the exposition clear, we have divided the
proof into five steps.

Given a (private coin) protocol $\prot$, let $\psim$ be its
$\ep$-simulation and $\Protx$ and $\Proty$ be the
corresponding estimates of the transcript $\Prot$ for
\partyx and \partyy, respectively.

%%%%%%%%%%%%%%%%%%%%
\subsection{From simulation to probability of error}

We first use a coupling argument to replace the
$\ep$-simulation condition with an $\ep$ probability of
error condition. Recall the maximal coupling lemma.

\begin{lemma}[{\bf Maximal Coupling Lemma }\cite{Str65}]
For any two distributions $\dP$ and $\dQ$ on the same set,
there exists a joint distribution $\bPP {XY}$ with $X\sim
\dP$ and $Y \sim \dQ$ such that
\[
\bPr{X \neq Y} = \ttlvrn {\dP}{\dQ}.
\]
\end{lemma}

\noindent Using the maximal coupling lemma, for each fixed
$x, y$ there exists a joint distribution $\bPP{\Prot \Protx
  \Proty | X=x, Y=y}$ such that
\[
\bPr{\Prot = \Protx = \Proty | X= x, Y =y} = 1 - \ttlvrn
    {\bPP{\Pprot| X=x, Y=y}}{\bPP{\Protx \Proty| X=x, Y=y}};
\]
Consequently,
\begin{align}
\bPr{\Prot = \Protx = \Proty} &= 1 -
\sum_{x,y}\bP{XY}{x,y}\ttlvrn {\bPP{\Pprot| X=x,
    Y=y}}{\bPP{\Protx \Proty| X=x, Y=y}} \nonumber \\ & = 1
- \ttlvrn{\bPP{\Pprot XY}}{\bPP{\Protx \Proty XY}} \nonumber
\\ &\geq 1- \ep.
\label{e:coupling_pe}
\end{align}
As pointed in footnote~\ref{foot:no_private_randomness}, we
restrict ourselves to public coin protocols $\psim$ using
shared public randomness $U$. For concreteness (and
convenience of proof), we define the joint distribution for
$(\Prot \Protx \Proty XYU)$ as
\begin{align}
\bPP{\Protx \Proty \Prot XYU} = \bPP{\Protx \Proty \Prot XY}
\bPP{U|\Protx \Proty XY}.
\label{e:joint_distribution_sim}
\end{align}
Note that the marginal $\bPP{\Protx \Proty XY U}$ remains as
in the original protocol.  In particular, $(X,Y)$ is jointly
independent of $U$.

%%%%%%%%%%%%%%%%%%%%
\subsection{From partial knowledge to omniscience}

As explained in the heuristic proof above, instead of
extracting a secret key from the common randomness generated
by the protocol $\psim$, we first use the data exchange
protocol of Theorem~\ref{t:DE_protocol} to make all the data
available to both the parties, which was termed {\it
  attaining omniscience}\footnote{Csisz\'ar and Narayan
  considered a multiterminal version of the data exchange
  problem in \cite{CsiNar04} and connected the minimum
  (amortized) rate of communication needed to the maximum
  (amortized) secret key rate.} in \cite{CsiNar04}.  In
particular, the parties run the protocol $\psim$ followed by
a data exchange protocol for $(X\Prot, Y\Prot)$ to recover
$(X,Y)$ at both the parties.  Once both the parties have
access to $(X,Y)$, they can extract a secret key from
$(X,Y)$ which will be used in the reduction in our final
step.

Formally, with the notations introduced in
Section~\ref{section:data-exchange}, let $\pdex$ be the data
exchange protocol of Theorem~\ref{t:DE_protocol} with $X$
and $Y$ replaced by $(X\Prot)$ and $(Y\Prot)$, respectively,
with $N_2$ and $\Delta_2$ denoting $N$ and $\Delta$,
respectively, and with $\lambda = \lamaxs 3$, $\lamin' =
\lamins 2$, $\lamax' = \lamaxs 2$.  Then, denoting by
$\cE_{\mathtt{error}}$ the error event for the protocol
$\pdex$ Theorem \ref{t:DE_protocol}(i) yields
\begin{align}
\bPr{\cE_{\mathtt{error}}\cap \cE_3^c} \leq \bPr{\cE_2} +
N_22^{-\xi},
\label{e:omniscience_pe}
\end{align}
where $\cE_2$ and $\cE_3$ are as in
\eqref{e:tail_prob}. Furthermore, for every realization
$(X,Y) \notin \cE_3$ the number possible transcripts $\Pdex$
is no more than
\begin{align}
2^{\hsump + \Delta_2 + \xi}.
\label{e:omniscience_NoT}
\end{align}

We seek to use $\pdex$ for recovering $Y$ and $X$,
respectively, at \partyx and \partyy by running $\pdex$
successively after $\psim$. However, $\psim$ yields
$X\Protx$ and $Y\Proty$ at \partyx and \partyy,
respectively, while the data exchange protocol $\pdex$
facilitates data exchange when the two parties observe
$X\Prot$ and $Y\Prot$. We can easily fix this gap using
\eqref{e:coupling_pe}.

Specifically, denote by $\hat X$ and $\hat Y$ the estimates
of $X$ and $Y$ formed at \partyy and \partyx in
$\pdex$. Note that $\pdex$ is a deterministic protocol and
$\hat X$ and $\hat Y$ are functions of $(X,Y, \Prot,
\Prot)$.  Denote by $\cA$ the set
\[
\cA = \{(\tau_\cX, \tau_\cY, \tau, x, y) : \tau_\cX =
\tau_\cY = \tau \}
\]
and by $\cB$ the set
\[
\cB = \{(\tau_\cX, \tau_\cY, \tau, x,y): \hat X(x, y, \tau,
\tau)=x, \hat Y(x, y, \tau, \tau)=y\},
\]
which is the same as $\cE_{\mathtt{error}}^{c}$ for
$\cE_{\mathtt{error}}$ in \eqref{e:omniscience_pe}.  Then,
by \eqref{e:coupling_pe} and \eqref{e:omniscience_pe}
\begin{align}
& \bPr{\{\hat X(X,Y,\Protx, \Proty) = X, \hat Y(X,Y, \Protx,
    \Proty) = Y\}\cap \cE_3^c} \nonumber \\ &\geq
  \bP{\Protx\Proty\Prot XY}{\cA\cap\cB\cap \cE_3^c}
  \nonumber \\ &\geq \bP{\Protx\Proty\Prot XY}{\cA} +
  \bPr{\cE_3^c} - \bP{\Protx\Proty\Prot XY}{\cB^c\cap
    \cE_3^c} - 1 \nonumber \\ &\geq 1 - \ep -\bPr{\cE_2} -
  \bPr{\cE_3} - N_22^{-\xi}.
\label{e:combined_DE}
\end{align}

%%%%%%%%%%%%%%%%%%%%
\subsection{From simulation to secret keys: A rough sketch of the reduction}
\label{s:rough_reduction}
The first step in our proof is to replace the simulation
condition \eqref{e:simulation} with the probability of error
condition \eqref{e:coupling_pe} for the joint distribution
$\bPP{\Prot_\cX \Prot_\cY \Prot XYU}$ in
\eqref{e:joint_distribution_sim}.

Next, we ``complete the common randomness,'' \ie, we
communicate more to facilitate the recovery of $Y$ and $X$
at \partyx and \partyy, respectively. To that end, upon
executing $\psim$, the parties run the data exchange
protocol $\pdex$ of Theorem~\ref{t:DE_protocol} for
$(X\Prot)$ and $(Y\Prot)$, with $(X, \Protx)$ and $(Y,
\Proty)$ in place of $(X\Prot)$ and $(Y\Prot)$,
respectively. Condition \eqref{e:coupling_pe} guarantees
that the combined protocol $(\psim,\pdex)$ recovers $Y$ and
$X$ at \partyx and \partyy with probability of error less
than $\ep$.

We now sketch our reduction argument. Consider the secret
key agreement for $X$ and $Y$ when the eavesdropper observes
$U$. By the independence of $(X,Y)$ and $U$, $S_{\eta}(XU,YU
|U) = S_{\eta}(X,Y)$, and further, the result of
\cite{TyaWat14} shows that $S_\eta(X,Y)$ is bounded above,
roughly, by the {\it mutual information density} $i(X\wedge
Y) = \log \bP{XY}{X,Y}/\bP X X\bP YY$, \ie,
\begin{align}
S_\eta(XU,YU|U) = S_\eta(X,Y) \lesssim i(X\wedge Y).
\label{e:rough_S_upper_bound}
\end{align}

On the other hand, we can generate a secret key using the
following protocol:
\begin{enumerate}
\item Run the combined protocol $(\psim, \pdex)$ to attain
  data exchange for $X$ and $Y$, resulting in a common
  randomness of size roughly $h(X,Y|U) = h(X,Y)$.

\item The data exchange protocol $\pdex$ for $(X\Prot)$ and
  $(Y\Prot)$ communicates roughly $\hsump$ bits for every
  fixed realization $(X,Y,\Prot)$. Thus, the combined
  protocol $(\psim, \pdex)$, which allows both the parties
  to recover $(X,Y)$, communicates no more than $|\psim| +
  \hsump$ bits for every fixed realization $(X,Y,\Prot)$.
  Using the leftover hash lemma, we can extract a secret key
  of rate roughly $h(X,Y) - |\psim| - \hsump$.
\end{enumerate}
The following approximate inequalities summarize our
reduction:
\begin{align}
S_\eta(XU,YU|U) &\ge S_\eta(X\hat Y,\hat X Y|\Psim\Pdex U)
\nonumber \\ &\gtrsim S_\eta(X\hat Y,\hat X Y|U) - |\psim|
- \hsump \nonumber \\ &\approx h(X,Y) - |\psim| - \hsump,
\label{e:rough_S_lower_bound}
\end{align}
where the first inequality is by Lemma~\ref{l:monotonicity}
and the the second by Lemma~\ref{l:leftover_hash}.

We note that the generation of secret keys from data
exchange was first proposed in \cite{CsiNar04} in an
amortized, IID setup and was shown to yield a secret key of
asymptotically optimal rate.

From \eqref{e:rough_S_upper_bound} and
\eqref{e:rough_S_lower_bound} it follows that
\[
|\psim| \gtrsim h(X,Y) - \hsump - i(X\wedge Y) =\icp,
\]
which is the required lower bound.

Clearly, the steps above are not precise. We have used
instantaneous communication and common randomness lengths in
our bounds whereas a formal treatment will require us to use
worst-case performance bounds for these
quantities. Unfortunately, such worst-case bounds do not
yield our desired lower bound for $\dcp$. To fill this gap,
we apply the arguments above not for the original
distribution $\bPP{\Protx\Proty\Prot XYU}$ but for the
conditional distribution $\bPP{\Protx\Proty\Prot XYU|\cE}$
where the event $\cE$ is carefully constructed in such a
manner that the aforementioned worst-case bounds are close
to instantaneous bounds for all realizations. Specifically,
$\cE$ is selected by appropriately slicing the spectrums of
the various information densities that appear in the
worst-case bounds.
%%%%%%%%%%%%%%%%%%%%

\subsection{From original to conditional probabilities: A Spectrum slicing argument}

To identify an appropriate critical event for conditioning,
we take recourse to spectrum slicing.  Specifically, we
identify an appropriate subset of intersection of slices of
spectrums (ii) and (iv) described in
Section~\ref{s:lower_bound}. For the combined protocol
$(\psim,\pdex)$ and the estimates $(\hat X, \hat Y)$ as
above, and $\lamins i, \lamaxs i$, $i=1,2,3$, as in
Section~\ref{s:lower_bound}, let
\begin{align*}
\Esim &= \{\Prot = \Protx = \Proty\}, \\ \EDE &=\{\hat
X(X,Y, \Protx, \Proty) = X,\, \hat Y(X,Y, \Protx, \Proty) =
Y\}, \\ \cE_\lambda &= \{\icp \geq \la\} \\ \cE_i^{(1)} &=
\{\lamins 1 + (i-1)\Delta_1 \le \hxy \le \lamins 1 +
i\Delta_1\}, \quad 1\le i \le N_1, \\ \cE_j^{(3)} &=
\{\lamins 3 + (j-1)\Delta_3 \le \hsum{ X\Prot}{ Y\Prot} \le
\lamins 3 + j\Delta_3\}, \quad 1\le j \le N_3,
\end{align*}
where
\[
N_1 = \frac{\lamaxs 1 - \lamins 1}{\Delta_1}\text{ and } N_3
= \frac{\lamaxs 3 - \lamins 3}{\Delta_3}.
\]
Note that $\cup_i \,\cE_i^{(1)} = \cE_1^{c}$ and $\cup_j\,
\cE_j^{(3)} = \cE_3^{c}$, where the events $\cE_1$ and
$\cE_3$ are as in \eqref{e:tail_prob}. Finally, define the
event $\cE_{ij}$ as follows:
\[
\cE_{ij} =
\Esim\cap\EDE\cap\cE_\lambda\cap\cE_i^{(1)}\cap\cE_j^{(3)},
\quad 1\le i \le N_1, 1\le j \le N_3.
\]
The next lemma says that (at least) one of the events
$\cE_{ij}$ has significant probability, and this particular
event will be used as the critical event in our proofs.
\begin{lemma}\label{l:critical_event}
There exists $i,j$ such that
\begin{align}
\bPr{\cE_{ij}} \ge \frac{\bPr{\cE_\lambda} - \ep - \eptail -
  N_22^{-\xi}}{N_1N_3} \ed \alpha.
\label{e:critical_event_prob}
\end{align}
\end{lemma}
{\it Proof.} Note that the event $\Esim\cap \EDE \cap
\cE_3^c$ is the same as the event $\cA\cap\cB \cap \cE_3^c$
of \eqref{e:combined_DE}. Therefore,
\begin{align*}
\bPr{\Esim\cap\EDE\cap\cE_\lambda\cap \cE_1^{c}\cap
  \cE_3^{c}} &\geq \bPr{\cE_\lambda} + \bPr{\Esim\cap
  \EDE\cap\cE_3^c} + \bPr{\cE_1^{c}} - 2 \\ &\geq
\bPr{\cE_\lambda} - \ep -\bPr{\cE_2} -\bPr{\cE_3} -
N_22^{-\xi} - \bPr{\cE_1} \\ &\geq \bPr{\cE_\lambda} - \ep -
\eptail- N_22^{-\xi},
\end{align*}
where the second inequality uses \eqref{e:combined_DE} and
and the third uses \eqref{e:tail_prob}. The proof is
completed upon noting that $\{\cE_{ij}\}_{i,j}$ constitutes
a partition of $\Esim\cap\EDE\cap\cE_\lambda\cap
\cE_1^{c}\cap \cE_3^{c}$ with $N_1N_3$ parts.  \qed

%%%%%%%%%%%%%%%%%%%%
\subsection{From simulation to secret keys: The formal reduction proof}

We are now in a position to complete the proof of our lower
bound.  For brevity, let $\cE$ denote the event $\cE_{ij}$
of Lemma~\ref{l:critical_event} satisfying $\bPr{\cE} \geq
\alpha$.

Our proof essentially formalizes the steps outlined in
Section~\ref{s:rough_reduction}, but for the conditional
distribution given $\cE$. With an abuse of notation, let
$S_\eta(X,Y|Z, \cE)$ denote the maximum length of an
$\eta$-secret key for two parties observing $X$ and $Y$, and
the eavesdropper's side information $Z$, when the
distribution of $(X,Y,Z)$ is given by $\bPP{XYZ|\cE}$. Then,
using Lemma~\ref{l:bound_beta_epsilon} with $\bQQ X = \bPP
X$ and $\bQQ Y = \bPP Y$, we get the following bound in
place of \eqref{e:rough_S_upper_bound}:
\begin{align}
S_{2\eta}(X,Y|\cE) &\le \gamma - \log \lc \bPr{\left\{(x,y):
  \log\frac{\bP{XY|\cE}{x,y}}{\bP X x\bP Yy} <
  \gamma\right\}\,\bigg|\, \cE} -3\eta\rc_+ + 2\log(1/\eta)
\nonumber \\ &\le \gamma - \log \lc \bPr{\left\{(x,y):
  \log\frac{\bP{XY}{x,y}}{\bP X x\bP Yy} < \gamma
  +\log\alpha\right\}\,\bigg|\, \cE} -3\eta\rc_+ +
2\log(1/\eta),
\label{e:S_upper_bound}
\end{align}
where $0<\eta<1/3$ is arbitrary and in the previous
inequality we have used
\[
\bP{XY|\cE}{x,y|\cE} \le \frac{\bP{XY}{x,y}}{\bPr{\cE}} \le
\frac{\bP{XY}{x,y}}{\alpha}.
\]

To replace \eqref{e:rough_S_lower_bound}, note that by
Lemma~\ref{l:monotonicity}
\begin{align}
S_{2\eta}(X,Y|\cE) &\ge S_{2\eta}(X\Psim\Pdex,Y\Psim\Pdex|U,
\Psim,\Pdex,\cE) \nonumber \\ &\ge S_{2\eta}(X{\hat Y},{\hat
  X}Y|U, \Psim,\Pdex,\cE).
\label{e:S_lower_bound1}
\end{align}
%where we have used the fact that we can always generate a secret key for $X$ and $Y$
%by first running the combined protocol $(\psim,\pdex)$, recovering estimates $\hat Y$ and
%$\hat X$ at \partyx and \partyy, and finally, 
%since now \partyx observes $(X, \hat Y)$, \partyy observes $(\hat X,Y)$ 
%and the eavesdropper observes $(U,\Psim,\Pdex)$, 
%generating a secret key attaining $S_{2\eta}(X{\hat Y},{\hat X}Y|U, \Psim\Pdex\cE)$.

Next, note that by \eqref{e:omniscience_NoT} the transcript
$\Psim\Pdex$ takes no more than $2^{|\psim|+\hsump
  +\Delta_2+\xi}$ values for every realization $(X,Y) \notin
\cE_3$.  However, when the event $\cE = \cE_{ij}$ holds,
$\hsump \leq \lamins 3 +j\Delta_3$. It follows by
Lemma~\ref{l:SK_relation} that
\begin{align}
&S_{2\eta}(X{\hat Y},{\hat X}Y|U \Psim\Pdex,\cE) \nonumber
  \\ &\geq S_\eta(X{\hat Y},{\hat X}Y|U,\cE) - |\psim| -
  \lamins 3 - j\Delta_3 -\Delta_2-\xi -2\log(1/2\eta).
\label{e:S_lower_bound2}
\end{align}

Also, since $\{X=\hat X, Y = \hat Y\}$ holds when we
condition on $\cE$,
\begin{align}
S_\eta(X{\hat Y},{\hat X}Y|U, \cE) &= S_\eta(XY, XY|U,\cE)
\nonumber \\ &\ge H_{\min}(\bPP{XYU|\cE}\mid U) -
2\log(1/2\eta),
\label{e:S_lower_bound3}
\end{align}
where the previous inequality is by the leftover hash
lemma. Furthermore, by using
\[
\bPP{XYU|\cE}(x,y,u) \leq \frac{\bPP{XYU}(x,y,u)}{\bPr{\cE}}
\leq \frac{\bPP{XYU}(x,y,u)}{\alpha}
\]
we can bound $H_{\min}(\bPP{XYU|\cE}\mid U)$ as follows:
\begin{align}
H_{\min}(\bPP{XYU|\cE}\mid U) &\geq \min_{x,y,u} - \log
\frac{\bP{XYU|\cE}{x,y,u}}{\bP U u} \nonumber \\ &\geq
\min_{x,y,u} - \log
\frac{\bP{XYU}{x,y,u}\indicator(\bP{XYU|\cE}{x,y,u}
  >0)}{\alpha\bP U u} \nonumber \\ &= \min_{x,y\in
  \cE^{(1)}_i} h_{\bPP{XY}}(x,y)+ \log \alpha \nonumber
\\ &\geq \lamins 1 + (i-1)\Delta_1 + \log \alpha.
\label{e:S_lower_bound4}
\end{align}

Thus, on combining
\eqref{e:S_lower_bound1}-\eqref{e:S_lower_bound4}, we get
\begin{align}
S_{2\eta}(X,Y|\cE) &\ge [\lamins 1 + (i-1)\Delta_1 - \lamins
  3 - j\Delta_3 + \log \alpha] -\Delta_2-\xi -4\log(1/2\eta)
- |\psim|.
\label{e:S_lower_bound}
\end{align}

To get a matching form of the upper bound
\eqref{e:S_upper_bound} for $S_{2\eta}(X,Y|\cE)$, note that
since\footnote{For clarity, we display the dependence of
  each information density on the underlying distribution in
  the remainder of this section.}
\[
-\ic_{\bPP{\Prot XY}}(\tau;x, y) = i_{\bPP{XY}}(x\wedge y)
- h_{\bPP{XY}}(x,y) + h_{\bPP{\Prot XY}}((x,\tau)\triangle
- (y,\tau)),
\]
and since under $\cE$
\begin{align*}
h_{\bPP{XY}}(x,y)&\leq \lamins 1 + i\Delta_1,
\\ h_{\bPP{XY\Prot}}((x,\tau)\triangle (y,\tau)) &\geq
\lamins 3 + (j-1)\Delta_3,
\end{align*}
it holds that
\begin{align*}
&\bPr{\left\{(x,y): i_{\bPP{XY}}(x\wedge y) < \gamma
    +\log\alpha\right\}\,\bigg|\, \cE} \\ &\geq
  \bPr{\left\{(x,y, \tau): -\ic_{\bPP{XY\Prot}}(x,y,\tau) <
    \gamma -\lamins 1 - i\Delta_1 + \lamins 3 +
    (j-1)\Delta_3 +\log\alpha\right\}\,\bigg|\, \cE}.
\end{align*}
On choosing
\[
\gamma = -\lambda +\lamins 1 + i\Delta_1 - \lamins 3 -
(j-1)\Delta_3 -\log\alpha,
\]
it follows from \eqref{e:S_upper_bound} that
\begin{align}
&S_{2\eta}(X,Y|\cE) \nonumber \\ &\le -\lambda +[\lamins 1 +
    i\Delta_1 - \lamins 3 - (j-1)\Delta_3 -\log\alpha] -
  \log \lc \bPr{\cE_{\lambda}\mid \cE} -3\eta\rc_+ +
  2\log(1/\eta) \nonumber \\ &\le -\lambda +[\lamins 1 +
    i\Delta_1 - \lamins 3 - (j-1)\Delta_3 -\log\alpha] -
  \log ( 1 -3\eta) + 2\log(1/\eta),
\label{e:S_upper_bound2}
\end{align}
where the equality holds since $\bPr{\cE_\lambda\mid \cE} =
1$.

Thus, by \eqref{e:S_lower_bound} and
\eqref{e:S_upper_bound2}, we get
\begin{align*}
|\psim| &\geq \lambda +2\log\alpha - \Delta_1 -\Delta_2
-\Delta_3 -\xi -6\log(1/\eta) + \log(1-3\eta) + 4
\\ &=\lambda + 2\log(\bPr{\cE_\lambda}-\ep -\eptail -\eta) -
2\log N_1 N_3 - (\Delta_1+\Delta_2+\Delta_3) -\log N_2
\\ &\hspace{8cm} -7\log (1/\eta) + \log(1-3\eta) + 4.
\end{align*}
where the equality holds for $\xi = -\log \eta + \log N_2$.
Note that the maximum value of the right-side above, when
maximized over $N_i$ and $\Delta_i$ under the constraint
$N_i\Delta_i = \Lambda_i$, $i=1,2,3$, occurs for $\Delta_1
=\Delta_3 =2$ and $\Delta_2 = 1$. Substituting this choice
of parameters, we get
\begin{align*}
|\psim| &\geq\lambda + 2\log(\bPr{\cE_\lambda}-\ep
-\eptail -\eta) - 2\log \Lambda_1 \Lambda_3 -\log \Lambda_2
-7\log (1/\eta) + \log(1-3\eta) + 3.  \\ &\geq \lambda -
2\log \Lambda_1 \Lambda_3 -\log \Lambda_2 -9\log (1/\eta) +
\log(1-3\eta) + 3.
\end{align*}
where the final inequality holds for every $\la$ such that
$\bPr{\cE_\la} \geq \ep + \eptail +2\eta$;
Theorem~\ref{t:lower_bound} follows upon maximizing the
right side-over all such $\la$.  \qed

%%%%%%%%%%%%%%%%%%%%%%%%%%%%%%%%%%%%%%%%%%%%

\section{Simulation Protocol and the Upper Bound}\label{s:simulation_protocols}

In this section, we formally present an $\ep$-simulation of
a given interactive protocol $\prot$ with bounded
rounds. For clarity, we build the simulation protocol in
steps.

%%%%%%%%
%%%%%%%%
\subsection{Sending $X$ using one-sided communication}

We start with the well-known Slepian-Wolf compression
problem \cite{SleWol73} where \partyx wants to transmit $X$
itself to \partyy using as few bits as possible. This
corresponds to simulating the deterministic protocol
$\Prot=\Prot_1 =X$. See
Remark~\ref{r:deterministic_protocols} in
Section~\ref{s:problem_statement} for a discussion on
simulation of deterministic protocols.
%Note that $\ep$-simulation
%corresponds to recovering $X$ with a probability of error $\ep$ since
%\begin{align*}
%\left\| \bPP{\hat{X}XY} - \bPP{XXY} \right\| = \mathbb{P}\left( \hat{X} \neq X \right),
%\end{align*}
%where $\hat{X}$ denotes the estimate of $X$ formed by \partyy using the protocol.
%%In fact, we consider the more demanding requirement of sending $X$ to $Y$, as opposed to simulation
%where we need not recover the actual $X$ by \partyy and require only a random variable that has almost the same
%distribution as $\Prot_1$.

For encoder, we use a hash function that is randomly chosen
from a 2-universal hash family $\hash_l(\cX)$; for decoder,
we use a kind of joint typical decoder \cite{CovTho06}. Let
the {\em typical set} $\TPXY$ be given by
\begin{align} \label{eq:typical-set}
\TPXY = \left\{ (x,y): \hPxy \le l - \gamma \right\}
\end{align} 
for a slack parameter $\gamma > 0$.  Our first protocol is
given below:
%%%% PROTOCOL %%%%
\begin{protocol}
\caption{Slepian-Wolf compression}
\label{p:slepian_wolf}
\KwIn{Observations $X$ and $Y$, uniform public randomness
  $U_{\Hash}$, and a parameter $l$} \KwOut{Estimate $\hat X$
  of $X$ at party 2} Both parties use $U_{\Hash}$ to select
$\hashfunc$ from $\hash_l(\cX)$\\ \partyx sends
$\Prot_{\Sim,1} = \hashfunc(X)$\\ \eIf{\partyy finds a
  unique $x\in \TPXY$ with hash value $\hashfunc(x) =
  \Prot_{\Sim,1}$} { set $\hat X = x$ } { protocol declares
  an error }
\end{protocol}
%%%%

The following result is from \cite{MiyKan95}, \cite[Lemma
  7.2.1]{Han03} (see, also, \cite{Kuz12}).
\begin{lemma}[{\bf Performance of Protocol~\ref{p:slepian_wolf}}]\label{l:slepian_wolf} 
For every $\gamma > 0$, the protocol above satisfies
\begin{align*} 
\bPr{X\neq \hat X} \leq \bP{XY}{\TPXY^c}+ 2^{-\gamma}.
\end{align*}
\end{lemma}
Essentially, the result above says that \partyx can send $X$
to \partyy with probability of error less than $\ep$ using
roughly as many bits as the $\ep$-tail of $\hPXY$.

In fact, the use of the typical set in
\eqref{eq:typical-set} is not crucial in Protocol
\ref{p:slepian_wolf} and its performance analysis: For a
given measure $\bQQ{XY}$, we can define another typical set
$\TQXY$ by replacing $\hPxy$ with $\hQxy$ in
\eqref{eq:typical-set} even though the underlying
distribution of $(X,Y)$ is $\bPP{XY}$.  Then, the error
probability is bounded as
\begin{align*} 
\bPr{X\neq \hat X} \leq \bP{XY}{\TQXY^c}+ 2^{-\gamma},
\end{align*}
which implies that $X$ can be sent by using roughly as many
bits as the $\ep$-tail of $\hQXY$ under $\bPP{XY}$.  This
modification simplifies our performance analysis of the more
involved protocols in the following sections.

%%%%%%%%
\subsection{Sending $X$ using interactive communication}
\label{subsec:interactive-SW}

Protocol~\ref{p:slepian_wolf} aims at minimizing the
worst-case communication length over all realization of
$(X,Y)$.  However, our goal here is to simulate a multiround
interactive protocol, and we need not account for
the worst-case communication length in each round.  Instead,
we shall optimize the worst-case communication length for
the combined interactive protocol. The protocol below is a
modification of Protocol~\ref{p:slepian_wolf} and uses
roughly $h(X|Y)$ bits for transmitting $X$ instead of its
$\ep$-tail.

The new protocol proceeds as the previous one but relies on
{\em spectrum-slicing} to adapt the length of communication
to the specific realization of $(X,Y)$: It increases the
size of the hash output gradually, starting with $\la_1 =
\la_{\min}$ and increasing the size $\Delta$-bits at a time
until either \partyy decodes $X$ or $\la_{\max}$ bits have
been sent.  After each transmission, \partyy sends either an
ACK-NACK feedback signal. The protocol stops when an ACK
symbol is received.
%The proposed protocol is illustrated in Figure~\ref{f:spectrum}.
%\begin{figure}[H]
%\centering
%\includegraphics[scale=0.5]{spectrum_slicing.pdf}
%\label{f:spectrum}
%\caption{Slicing the spectrum of $\bPP {X\mid Y}$}
%\end{figure}

Fix an auxiliary distribution $\bQQ{XY}$.  For
$\lambdaQXYmin, \lambdaQXYmax, \DeltaQXY >0$ with
$\lambdaQXYmax > \lambdaQXYmin$, let
\begin{align*}
\NQXY = \frac{\lambdaQXYmax - \lambdaQXYmin}{\DeltaQXY},
\end{align*}
and
\begin{align*}
\la^{(i)}_{\bQQ{X|Y}} = \lambdaQXYmin + (i-1) \DeltaQXY,
\quad 1 \le i \le \NQXY.
\end{align*}
Further, let
\begin{align} \label{eq:tail-set-interactive-SW}
 \TQXYzero := \left\{ (x,y) \mid \hQxy \geq \lambdaQXYmax
 \text{ or } \hQxy < \lambdaQXYmin \right\},
\end{align}
and for $1\le i \le \NQXY$, let $\TQXYi$ denote the $i$th
slice of the spectrum given by
\begin{align*}
\TQXYi = \left\{(x,y) \mid \lambdaQXYi \leq \hQxy <
\lambdaQXYi + \DeltaQXY \right\}.
\end{align*}
Note that $\TQXYzero$ corresponds to $\cT_{\bQQ{X|Y}}^c$ in
the previous section and will be counted as an error event.

%%%% PROTOCOL %%%%
\begin{protocol}[h]
\caption{Interactive Slepian-Wolf compression}
\label{p:slepian_wolf_interactive}
\KwIn{Observations $X$ and $Y$ with distribution $\bPP{XY}$,
  uniform public randomness $U_{\Hash}$, auxiliary
  distribution $\bQQ{XY}$, and parameters $\gamma$,
  $\lambdaQXYmin$, $\DeltaQXY$, $\NQXY$, and $l$}
\KwOut{Estimate $\hat X$ of $X$ at party 2} Both parties use
$U_{\Hash}$ to select $\hashfunc_1$ from
$\hash_{l}(\cX)$\\ \partyx sends $\Prot_{\Sim,1} =
\hashfunc_1(X)$\\ \eIf{\partyy finds a unique $x\in
  \TQXYone$ with hash value $\hashfunc_1(x) =
  \Prot_{\Sim,1}$} { set $\hat X = x$\newline send back
  $\Prot_{\Sim,2} = \text {ACK}$ } { send back
  $\Prot_{\Sim,2} = \text{NACK}$ } \While{ $2\le i\le \NQXY$
  and party 2 did not send an ACK} { Both parties use
  $U_{\Hash}$ to select $\hashfunc_i$ from
  $\hash_{\DeltaQXY}(\cX)$, independent of $\hashfunc_1,
  ..., \hashfunc_{i-1}$\\ \partyx sends $\Prot_{\Sim,2i-1} =
  \hashfunc_i(X)$\\ \eIf{\partyy finds a unique $x\in
    \TQXYi$ with hash value $\hashfunc_j(x) =
    \Prot_{\Sim,2j-1},\, \forall\, 1\le j\le i$} { set $\hat
    X = x$\newline send back $\Prot_{\Sim,2i} = \text{ACK}$
  } { \eIf{More than one such $x$ found} { protocol declares
      an error } { send back $\Prot_{\Sim,2i} = \text{NACK}$
  } } Reset $i \rightarrow i+1$ } \If{No $\hat X$ found at
  party 2} { Protocol declares an error }
\end{protocol}
%%%%%

Our protocol is described in
Protocol~\ref{p:slepian_wolf_interactive}. For every
$(x,y)\in \TQXYi$, $1\le i \le \NQXY$, the following lemma
provides a bound on the error.
%%%
\begin{lemma}[{\bf Performance of Protocol~\ref{p:slepian_wolf_interactive}}]\label{l:slepian_wolf_interactive}
For $(x,y)\in \TQXYi$, $1\le i \le \NQXY$, denoting by $\hat
X = \hat X(x,y)$ the estimate of $x$ at \partyy at the end
of the protocol (with the convention that $\hat X =
\emptyset$ if an error is declared),
Protocol~\ref{p:slepian_wolf_interactive} sends at most
$(l+(i-1)\DeltaQXY + i)$ bits and has probability of error
bounded above as follows:
\[
\bPr{\hat X \neq x \mid X=x, Y=y} \leq i2^{\lambdaQXYmin
  +\DeltaQXY - l}.
\]
\end{lemma}
%%%
\begin{proof}
Since $(x,y)\in \TQXYi$, an error occurs if there exists a
$\hat{x}\neq x$ such that $(\hat x,y)\in \TQXYj$ and
$\Prot_{\Sim,2k-1} = \hashfunc_{2k-1}(\hat x)$ for $1 \leq k
\leq j$ for some $j\leq i$. Therefore, the probability of
error is bounded above as
\begin{align*}
\bPr{\hat X \neq x \mid X=x, Y=y}&\leq
\sum_{j=1}^{i}\sum_{\hat x\neq x} \bPr{\hashfunc_{2k-1}(x) =
  \hashfunc_{2k-1}(\hat x),\, \forall\, 1\leq k \leq j}
\mathbbm{1}\left((\hat x, y)\in \TQXYj \right) \\&\leq
\sum_{j=1}^{i}\sum_{\hat x\neq x}
\frac{1}{2^{l+(j-1)\DeltaQXY}}\mathbbm{1}\left((\hat x,
y)\in \TQXYj \right) \\&= \sum_{j=1}^{i}\sum_{\hat x\neq
  x}\frac{1}{2^{l+(j-1)\DeltaQXY}} \left| \left\{\hat x \mid
(\hat x, y)\in \TQXYj \right\} \right| \\ &\leq
i2^{\lambdaQXYmin + \DeltaQXY - l},
\end{align*}
where the first inequality follows from the union bound, the
second inequality follows from the property of 2-universal
hash family, and the third inequality follows from the fact
that
\[
|\{\hat x \mid (\hat x, y)\in \TQXYj\}| \leq
2^{\lambdaQXYj+\DeltaQXY}.
\]
Note that the protocol sends $l$ bits in the first
transmission, and $\DeltaQXY$ bits and $1$-bit feedback in
every subsequence transmission.  Therefore, no more than
$(l+(i-1)\DeltaQXY + i)$ bits are sent.
\end{proof}
\begin{corollary}\label{c:protocol2}
Protocol~\ref{p:slepian_wolf_interactive} with $l =
\lambdaQXYmin+\DeltaQXY + \gamma$ sends at most $(\hQXY +
\DeltaQXY + \gamma +\NQXY)$ bits when the observations
are\footnote{ When $\hQXY < \lambdaQXYmin$,
  Protocol~\ref{p:slepian_wolf_interactive} may transmit
  more than $(\hQXY + \DeltaQXY + \gamma +\NQXY)$ bits.}
$(X,Y) \notin \TQXYzero$, and has probability of error less
than
\[
\bPr{\hat X \neq X} \leq \bPr{(X,Y)\in \TQXYzero} + \NQXY
2^{-\gamma}.
\]
\end{corollary}

%%%%%%%%%%%%%%%%%%%%%%%%%%%%%%%%%%%%%%%
\subsection{Simulation of $\Prot_1$ using interactive communication}

We now proceed to simulating the first round of our given
interactive protocol $\prot$.  Note that using
Protocol~\ref{p:slepian_wolf_interactive}, we can send
$\Prot_1$ using roughly $h(\Prot_1|Y)$ bits. This protocol
uses a public randomness $U_{\Hash}$ only to choose hash
functions, which is convenient for our probability of error
analysis, and can be easily derandomized.  We now present a
scheme which uses another independent portion of public
randomness $U_{\Sim}$ to reduce the rate of the
communication further. However, the scheme will only allow
the parties to simulate $\Prot_1$ (rather than recover it
with small probability of error) and cannot be derandomized.

Specifically, our next protocol uses $X$ and
$U=(U_{\Hash},U_{\Sim})$ to simulate $\Prot_1$ in such a
manner that $U_{\Sim}$ can be treated, in effect, as a
portion of the communication used in
Protocol~\ref{p:slepian_wolf_interactive}. Note that since
$U_\Sim$ is independent of $(X,Y)$, the portion of
communication which is equivalent to $U_\Sim$ must as well
be almost independent of $(X,Y)$. Such a portion can be
guaranteed by noting that the communication used in
Protocol~\ref{p:slepian_wolf_interactive} is simply a random
hash of $\Prot_1$ drawn from a 2-universal family, and
therefore, its appropriately small portion can have the
desired independence property by the leftover hash lemma.
In fact, since the Markov condition $\Prot_1 \mc X \mc Y$
holds, it suffices guarantee the independent of $X$ instead
of $(X,Y)$.

%%%% PROTOCOL %%%%
\begin{protocol}[h]
\caption{Simulation of $\Prot_1$}
\label{p:simulation_one_message}
\KwIn{Observations $X$ and $Y$ with distribution $\bPP{XY}$,
  uniform public randomness $U=(U_{\Hash},U_{\Sim})$,
  auxiliary distribution $\bQQ{\Prot_1Y}$, and parameters
  $\gamma$, $\lambdaQPiOneYmin$, $\DeltaQPiOneY$,
  $\NQPiOneY$ and $k$} \KwOut{Estimates $\Prot_{1\cX}$ and
  $\Prot_{1\cY}$ of $\Prot_1$} {\bf 1.} Two parties share
$k$ random bits $U_{\Sim}$ and an $h$ chosen from ${\cal
  H}_k(\mathrm{supp}(\Prot_1))$ using $U_{\Hash}$\\ {\bf 2.}
\partyx generates a sample $\Prot_{1\cX}$ using $\bP{\Prot_1
  | X \hashfunc(\Prot_1)}{\cdot | X, U_{\Sim}}$ \\ {\bf 3.}
Parties use Protocol~\ref{p:slepian_wolf_interactive} with
auxiliary distribution $\bQQ{\Prot_1 Y}$, and parameters
$\gamma$, $\lambdaQPiOneYmin$, $\DeltaQPiOneY$, $\NQPiOneY$,
and $l= \lambdaQPiOneYmin + \DeltaQPiOneY + \gamma$ to send
$\Prot_{1\cX}$ to \partyy by treating $U_{\Sim}$ as the
first $k$ bits of communication obtained via the hash
function $\hashfunc$
\end{protocol}
Our simulation protocol is described in
Protocol~\ref{p:simulation_one_message}.  Let the quantities
such as $\lambdaQPiOneYmin, \DeltaQPiOneY$, and $\NQPiOneY$
be defined analogously to the corresponding quantities in
Section \ref{subsec:interactive-SW} with $\Prot_1$ replacing
$X$.  The following lemma provides a bound on the simulation
error for Protocol~\ref{p:simulation_one_message}.
\begin{lemma}[{\bf Performance of Protocol~\ref{p:simulation_one_message}}]\label{l:simulation_one_message}
Protocol~\ref{p:simulation_one_message} sends at most
\begin{align*}
\left( \hQPiOneYPioneX + \DeltaQPiOneY + \NQPiOneY + \gamma
- k \right)_+
\end{align*}
bits when $(\Prot_{1\cX},Y) \notin \TQPiOneYzero$, and has
simulation error
\begin{align*}
\ttlvrn{ \bPP{\Prot_{1\cX}\Prot_{1\cY}XY}}{
  \bPP{\Prot_1\Prot_1 XY}} \le \bPr{(\Prot_1,Y) \in
  \TQPiOneYzero} + \NQPiOneY 2^{-\gamma} + \frac{1}{2}
\sqrt{2^{k - H_{\min}(\bPP{\Prot_1X}|\bQQ{X})}}
\end{align*}
for any auxiliary distribution $\bQQ{X}$ on $\cX$.
\end{lemma}
%%%
\begin{proof}
Consider the following simple protocol for simulating
$\Prot_1$ at \partyy:
\begin{enumerate}
\item \partyx generates a sample $\Prot_{1}$ using
  $\bP{\Prot_1 | X}{\cdot | X}$.

\item Both parties use
  Protocol~\ref{p:slepian_wolf_interactive} with auxiliary
  distribution $\bQQ{\Prot_1 Y}$, and parameters $\gamma$,
  $\lambdaQPiOneYmin$, $\DeltaQPiOneY$, $\NQPiOneY$, and $l=
  \lambdaQPiOneYmin + \DeltaQPiOneY + \gamma$ to generate an
  estimate $\hat{\Prot}_1$ of $\Prot_{1}$ at \partyy.
\end{enumerate}
In this protocol, $l_{\mathsf{wst}} = \lambdaQPiOneYmin +
\NQPiOneY \DeltaQPiOneY + \gamma$ bits of hash values will
be sent for the worst $(\Prot_1,Y)$.  We divide these
$l_{\mathsf{wst}}$ hash values into two parts, the fist $k$
bits and the last $l_{\mathsf{wst}} - k$ bits; let
$\hashfunc$ and $\hashfunc^\prime$, respectively, denote the
hash function producing the first and the second
parts. Protocol~\ref{p:simulation_one_message} replaces, in
effect, $\hashfunc$ with shared randomness $U_\Sim$ for an
appropriately chosen value of $k$.

Note that the joint distribution of the random variables
involved in the simple protocol above
satisfies\footnote{When the protocol terminate before
  $\NQPiOneY$th round, a part of
  $(\hashfunc(\Prot_1),\hashfunc^\prime(\Prot_1))$ may not
  be sent.}
\begin{align} 
& \bPP{\hashfunc(\Prot_1) \hashfunc^\prime(\Prot_1) \Prot_1
    \hat{\Prot}_1 XY}(v,v^\prime,\protr,\hat{\protr},x,y)
  \nonumber \\ &= \bPP{\hashfunc(\Prot_1) X}(v,x)
  \bPP{\Prot_1|X \hashfunc(\Prot_1)}(\protr | x,v)
  \bPP{\hashfunc^\prime(\Prot_1)|\Prot_1}(v^\prime | \protr)
  \bPP{Y|X}(y|x) \bPP{\hat{\Prot}_1| \hashfunc(\Prot_1)
    \hashfunc^\prime(\Prot_1) \Prot_1 X Y}(\hat{\protr} |
  v,v^\prime,\protr,x,y).
   \label{eq:factorization-virtual}
\end{align}
Note that the simple protocol above is deterministic and
therefore by Remark~\ref{r:deterministic_protocols}
\begin{align}
\ttlvrn{ \bPP{\hashfunc(\Prot_1) \hashfunc^\prime(\Prot_1)
    \Prot_1 \hat{\Prot}_1 XY}}{ \bPP{\hashfunc(\Prot_1)
    \hashfunc^\prime(\Prot_1) \Prot_1 \Prot_1 XY}} &= \bPr{
  \Prot_1 \neq \hat{\Prot}_1} \nonumber \\ &\le
\bPr{(\Prot_1,Y) \in \TQPiOneYzero} + \NQPiOneY 2^{-\gamma},
 \label{eq:simulation-error-virtual}
\end{align}
where the inequality is by Corollary \ref{c:protocol2}.

On the other hand, the joint distribution of random
variables involved in
Protocol~\ref{p:simulation_one_message} can be factorized as
\begin{align}
 & \bPP{U_{\Sim} \hashfunc^\prime(\Prot_{1\cX}) \Prot_{1\cX}
    \Prot_{1\cY} XY}(u,u^\prime,\protr,\hat{\protr},x,y)
  \nonumber \\ &= \bPP{U_{\Sim}}(u) \bPP{X}(x)
  \bPP{\Prot_1|X \hashfunc(\Prot_1)}(\protr |x, u)
  \bPP{\hashfunc^\prime(\Prot_1)|\Prot_1}(u^\prime |\protr)
  \bPP{Y|X}(y|x) \bPP{\hat{\Prot}_1| \hashfunc(\Prot_1)
    \hashfunc^\prime(\Prot_1) \Prot_1 X Y}(\hat{\protr} |
  u,u^\prime,\protr,x,y).
 \label{eq:factorization-real}
\end{align}
Therefore, the simulation error for
Protocol~\ref{p:simulation_one_message} is bounded as
\begin{align*}
& \ttlvrn{
    \bPP{\Prot_{1\cX}\Prot_{1\cY}XY}}{\bPP{\Prot_1\Prot_1XY}}
  \nonumber \\ &\le \ttlvrn{ \bPP{U_{\Sim}
      \hashfunc^\prime(\Prot_1) \Prot_{1\cX} \Prot_{1\cY}
      XY}}{ \bPP{\hashfunc(\Prot_1)
      \hashfunc^\prime(\Prot_1) \Prot_1 \Prot_1 XY}}
  \nonumber \\ &\le \ttlvrn{ \bPP{U_{\Sim}
      \hashfunc^\prime(\Prot_1) \Prot_{1\cX} \Prot_{1\cY}
      XY}}{ \bPP{\hashfunc(\Prot_1)
      \hashfunc^\prime(\Prot_1) \Prot_1 \hat{\Prot}_1 XY}} +
  \ttlvrn{ \bPP{\hashfunc(\Prot_1) \hashfunc^\prime(\Prot_1)
      \Prot_1 \hat{\Prot}_1 XY}}{ \bPP{\hashfunc(\Prot_1)
      \hashfunc^\prime(\Prot_1) \Prot_1 \Prot_1 XY}}
  \nonumber \\ &= \ttlvrn{ \bPP{U_{\Sim}} \bPP{X}}{
    \bPP{\hashfunc(\Prot_1) X}} + \ttlvrn{
    \bPP{\hashfunc(\Prot_1) \hashfunc^\prime(\Prot_1)
      \Prot_1 \hat{\Prot}_1 XY}}{ \bPP{\hashfunc(\Prot_1)
      \hashfunc^\prime(\Prot_1) \Prot_1 \Prot_1 XY}}
  \nonumber \\ &\le \ttlvrn{ \bPP{U_{\Sim}} \bPP{X}}{
    \bPP{\hashfunc(\Prot_1) X}} + \bPr{(\Prot_1,Y) \in
    \TQPiOneYzero} + \NQPiOneY 2^{-\gamma},
\end{align*}
where the first inequality is by the monotonicity of 
$\ttlvrn{\cdot}{\cdot }$, the second inequality is by the triangular
inequality, the equality is by the fact that replacing
$\bPP{U_{\Sim}} \bPP{X}$ with $\bPP{\hashfunc(\Prot_1)X}$ is
the only difference between the factorizations in
\eqref{eq:factorization-real} and
\eqref{eq:factorization-virtual}, and the final inequality
is by \eqref{eq:simulation-error-virtual}. The desired bound
on simulation error for
Protocol~\ref{p:simulation_one_message} follows by using
Lemma \ref{l:leftover_hash} to get
\begin{align*}
\ttlvrn{ \bPP{U_{\Sim}} \bPP{X}}{ \bPP{\hashfunc(\Prot_1)
    X}} \le \frac{1}{2} \sqrt{2^{k -
    H_{\min}(\bPP{\Prot_1X}|\bQQ{X})}}.
\end{align*}
Since Protocol~\ref{p:simulation_one_message} uses shared
randomness $U_{\Sim}$ instead of sending
$\hashfunc(\Prot_1)$, it communicates $k$ fewer bits in
comparison with the simple protocol above, which completes
the proof.
\end{proof}

%%%%%%%%%%%%%%%%%%%%%%%%%%%%%%%%%%%%%%%%
\subsection{Improved simulation of $\Prot_1$}

In Protocol~\ref{p:simulation_one_message} we were able to
reduce the communication by roughly
$H_{\min}(\bPP{\Prot_1X}|\bQQ{X})$ bits by simulating a
$\Prot_1$ such that if we use
Protocol~\ref{p:slepian_wolf_interactive} for sending
$\Prot_1$ to \partyy, a portion of the required
communication can be treated as shared public
randomness. However, this is the least reduction in
communication we can obtain in the worst-case.  In this
section, we slice the spectrum of $\hPPiOneXPione$ to obtain
an instantaneous reduction of roughly $\hPPiOneXPione$ bits.

%We define $\lambdaPPiOneXmin$, $\DeltaPPiOneX$, $\NPPiOneX$, and etc. in the same manner as in
%Section \ref{subsec:interactive-SW}. 
Denote by $J$ a random variable which takes the value $j \in
\{0,1,\ldots,\NPPiOneX\}$ if $(\Prot_1,X) \in
\TPPiOneXj$. In our modified protocol, \partyx first samples
$J$ and sends it to \partyy.  Then, they proceed with
Protocol~\ref{p:simulation_one_message} for
$\bPP{\Prot_1XY|J=j}$ by selecting $k$ to be less than
$H_{\min}(\bPP{\Prot_1X|J=j}|\bQQ{X})$ for an appropriately
chosen $\bQQ{X}$.  Let $\cJ_{\mathsf{g}}$ be the set of
"good" indices $j > 0$ with
\begin{align*}
\bP{J}{j} \ge \frac{1}{\NPPiOneX^2};
\end{align*}
it holds that
\begin{align*}
\bP{J}{\cJ_{\mathsf{g}}^c} < \bPr{(\Prot_1,X) \in
  \TPPiOneXzero} + \frac{1}{\NPPiOneX}.
\end{align*}

Note that for $j \in \cJ_{\mathsf{g}}$, with $\bQQ{X} =
\bPP{X}$, we have
\begin{align*}
H_{\min}(\bPP{\Prot_1X|J=j} | \bPP{X}) &= \min_{\protr,x}
-\log \frac{\bP{\Prot_1X|J}{\protr,x|j}}{\bP{X}{x}}
\nonumber \\ &= \min_{\protr,x} -\log
\frac{\bP{\Prot_1|X}{\protr|x}}{\bP{J}{j}} \nonumber \\ &\ge
\lambdaPPiOneXmin + (j-1) \DeltaPPiOneX - 2 \log \NPPiOneX.
\end{align*}

%%%% PROTOCOL %%%%%
\begin{protocol}[h]
\caption{Improved simulation of $\Prot_1$}
\label{p:simulation_one_message2}
\KwIn{Observations $X$ and $Y$ with distribution $\bPP{XY}$,
  uniform public randomness $U = (U_{\Hash},U_{\Sim})$, and
  parameters $\lambdaPPiOneYmin$, $\DeltaPPiOneY$,
  $\NPPiOneY$, $\lambdaPPiOneXmin$, $\DeltaPPiOneX$,
  $\NPPiOneX$, and $\gamma$} \KwOut{Estimates $\Prot_{1\cX}$
  and $\Prot_{1\cY}$ of $\Prot_1$} \partyx generate $J \sim
\bPP{J | X}(\cdot | X)$, and send it to \partyy.\\ \eIf{$J
  =j \in \cJ_g$} { Parties use
  Protocol~\ref{p:simulation_one_message} with auxiliary
  distribution $\bPP{\Prot_1 Y}$, parameters $\gamma$,
  $\lambdaPPiOneYmin$, $\DeltaPPiOneY$, $\NPPiOneY$, and $k=
  \lambdaPPiOneXmin + (j-1)\DeltaPPiOneX - 2 \log \NPPiOneX
  - 2\gamma + 2$ to simulate $\Prot_{1\cX}$ and
  $\Prot_{1\cY}$ for the distribution $\bPP{\Prot_1 XY |
    J=j}$ } { protocol declares an error }
\end{protocol}
%%%%

Our modified simulation protocol is described in
Protocol~\ref{p:simulation_one_message2}.  The following
lemma provides a bound on the simulation error.
\begin{lemma}[{\bf Performance of Protocol~\ref{p:simulation_one_message2}}]\label{l:simulation_one_message2}
Protocol~\ref{p:simulation_one_message2} sends at most
\begin{align*}
\left( \hPPiOneYPioneX - \hPPiOneXPioneX + \NPPiOneY + 3
\log \NPPiOneX + \DeltaPPiOneY + \DeltaPPiOneX + 3\gamma
\right)_+
\end{align*}
bits when $(\Prot_{1\cX},Y) \notin \TPPiOneYzero$, and has
simulation error
\begin{align*}
& \ttlvrn{ \bPP{\Prot_{1\cX}\Prot_{1\cY}XY}}{
    \bPP{\Prot_1\Prot_1XY}} \\ & \le \bPr{ (\Prot_1,Y) \in
    \TPPiOneYzero } + \bPr{ (\Prot_1,X) \in \TPPiOneXzero }
  + \left( \NPPiOneY + 1\right) 2^{-\gamma} +
  \frac{1}{\NPPiOneX}.
\end{align*}
\end{lemma}
%%%
\begin{proof}
First, we have
\begin{align*} 
& \ttlvrn{ \bPP{\Prot_{1\cX}\Prot_{1\cY}XY}}{
    \bPP{\Prot_1\Prot_1XY}} \nonumber \\ &\le \ttlvrn{
    \bPP{\Prot_{1\cX}\Prot_{1\cY}XYJ}}{
    \bPP{\Prot_1\Prot_1XYJ}} \nonumber \\ &= \sum_j
  \bPP{J}(j) \ttlvrn{
    \bPP{\Prot_{1\cX}\Prot_{1\cY}XY|J=j}}{\bPP{\Prot_1\Prot_1XY|J=j}}
  \nonumber \\ &\le \sum_{j \in \cJ_{\mathsf{g}}} \bPP{J}(j)
  \ttlvrn{
    \bPP{\Prot_{1\cX}\Prot_{1\cY}XY|J=j}}{\bPP{\Prot_1\Prot_1XY|J=j}}
  + \bP{J}{\cJ_{\mathsf{g}}^c} \nonumber \\ &\le \sum_{j \in
    \cJ_{\mathsf{g}}} \bPP{J}(j) \ttlvrn{
    \bPP{\Prot_{1\cX}\Prot_{1\cY}XY|J=j}}{
    \bPP{\Prot_1\Prot_1XY|J=j}} + \bPr{ (\Prot_1,X) \in
    \TPPiOneXzero } + \frac{1}{\NPPiOneX}.
\end{align*}
Then, we apply Lemma \ref{l:simulation_one_message} with
$\bQQ{X} = \bPP{X}$ for each $j \in \cJ_{\mathsf{g}}$, and
get
\begin{align}
& \ttlvrn{ \bPP{\Prot_{1\cX}\Prot_{1\cY}XY|J=j}}{
    \bPP{\Prot_1\Prot_1XY|J=j}} \nonumber \\ &\le \bPr{
    (\Prot_1,Y) \in \TPPiOneYzero \mid J = j } + \NPPiOneY
  2^{-\gamma} + \frac{1}{2} \sqrt{ 2^{k -
      H_{\min}(\bPP{\Prot_1X|J=j} | \bPP{X})} } \nonumber
  \\ &\le \bPr{ (\Prot_1,Y) \in \TPPiOneYzero \mid J = j } +
  \left(\NPPiOneY + 1\right) 2^{-\gamma}.
\end{align}
Thus, we have the desired bound on simulation error.

Next, we prove the claimed bound on the number of bits sent
by the protocol.  By Lemma \ref{l:simulation_one_message},
the fact that $J$ can be sent by using at most $\log
\NPPiOneX + 1$ bits and the choice of $k$ in
Protocol~\ref{p:simulation_one_message2}, for $J=j$ the
protocol above communicates at most
\begin{align*}
&\hQPiOneYPioneX + \DeltaQPiOneY + \NQPiOneY + \gamma + \log
  \NPPiOneX + 2 - k \\ &\leq \hQPiOneYPioneX -
  \lambdaPPiOneXmin - (j-1)\DeltaPPiOneX + \DeltaQPiOneY +
  \NQPiOneY + 3\log \NPPiOneX + 3\gamma.  \\ &\leq
  \hQPiOneYPioneX - \hPPiOneXPioneX + \DeltaPPiOneX +
  \DeltaQPiOneY + \NQPiOneY + 3\log \NPPiOneX + 3\gamma,
\end{align*}
where the previous inequality holds since for $\Prot_{1\cX}$
generated by $\bPP{\Prot_1|X \hashfunc(\Prot_1)J}(\cdot |
X,U_{\Sim},j)$
\begin{align*}
\lambdaPPiOneXmin + j \DeltaPPiOneX \ge \hPPiOneXPioneX,
\end{align*}
for each $j \in \cJ_{\mathsf{g}}$.
\end{proof}

%%%%%%%%%%%%%%%%%%%%%%%%%%%%%%%%%%%%%%%%%%%%%
\subsection{Simulation of $\Prot$}\label{s:simulation_final}

We are now in a position to describe our complete simulation
protocol.  Consider an interactive protocol $\prot$ with
maximum number of rounds $r_{\max} = d < \infty$.  We simply
apply Protocol~\ref{p:simulation_one_message2} for each
round $\Prot_t$ of $\Prot$.  Our overall simulation protocol
is described in Protocol~\ref{p:simulation_general}.  In
each round we use Protocol~\ref{p:simulation_one_message2}
assuming that the simulation up to the previous round has
succeeded, where, for the rounds with even numbers, we use
Protocol~\ref{p:simulation_one_message2} by interchanging
the role of \partyx and \partyy.

%%%%% PROTOCOL %%%%%
\begin{protocol}[h]
\caption{Simulation of $\Prot$}
\label{p:simulation_general}
\KwIn{Observations $X$ and $Y$ with distribution $\bPP{XY}$,
  uniform public randomness $U = (U_{t,\Hash},U_{t,\Sim}:
  t=1,\ldots,d)$, and parameters $\lambdaPPitXmin$,
  $\DeltaPPitX$, $\NPPitX$, $\lambdaPPitYmin$,
  $\DeltaPPitY$, $\NPPitY$ for $t=1,\ldots,d$ and $\gamma$.}
\KwOut{Estimates $\Prot_{\cX}$ and $\Prot_{\cY}$ of $\Prot$}
\While{Total communication is less than $l_{\max}$ bits, and
  simulation not ended} { \partyx and \partyy, respectively,
  use estimates $\Prot_{\cX}^{t-1}$ and $\Prot_{\cY}^{t-1}$
  for $\Prot^{t-1}$ \; Parties use
  Protocol~\ref{p:simulation_one_message2} for simulating
  $\bPP{\Prot_t (X\Prot^{t-1}) (Y \Prot^{t-1})}$ with
  parameters $\lambdaPPitXmin$, $\DeltaPPitX$, $\NPPitX$,
  $\lambdaPPitYmin$, $\DeltaPPitY$, $\NPPitY$ and $\gamma$
  \; Update $t\rightarrow t+1$ } \If{Total communication
  exceeds $l_{\max}$ bits} {Declare an error}
\end{protocol}
%%%%

The following lemma provides a bound on the simulation
error.
%%%
\begin{lemma}[{\bf Performance of Protocol~\ref{p:simulation_general}}]\label{l:simulation_general}
Protocol~\ref{p:simulation_general} sends at most $l_{\max}$
bits, and has simulation error
\begin{align*}
& \ttlvrn{ \bPP{\Prot_{\cX} \Prot_{\cY} XY}}{ \bPP{\Prot
      \Prot XY}} \\ &\le \bPr{ \icp + \sum_{t=1}^d \delta_t
    > l_{\max} } \\ &~~~ + \sum_{t=1}^d \bigg[ 4 \bPr{
      (\Prot_t,(Y,\Prot^{t-1})) \in \TPPitYzero } + 4 \bPr{
      (\Prot_t,(X,\Prot^{t-1})) \in \TPPitXzero } \\ &~~~ +
    3 \left( \NPPitY + \NPPitX + 2 \right) 2^{-\gamma} +
    \frac{3}{\NPPitX} + \frac{3}{\NPPitY} \bigg],
\end{align*} 
where
\begin{align}
\delta_t = \left\{
\begin{array}{ll}
\NPPitY + 3 \log \NPPitX +\DeltaPPitY + \DeltaPPitX + 3
\gamma & \mbox{odd } t \\ \NPPitX + 3 \log \NPPitY
+\DeltaPPitX + \DeltaPPitY + 3 \gamma & \mbox{even }t
\end{array}
\right..
\label{e:delta_t}
\end{align}
\end{lemma}
%%%
\begin{remark}\label{r:fudge_upper_bound}
The fudge parameters $\ep^\prime$ and $\la^\prime$ are given by
\begin{align*}
\ep^\prime &= \sum_{t=1}^d \bigg[ 4 \bPr{
    (\Prot_t,(Y,\Prot^{t-1})) \in \TPPitYzero } + 4 \bPr{
    (\Prot_t,(X,\Prot^{t-1})) \in \TPPitXzero } \\ &~~~ + 3
  \left( \NPPitY + \NPPitX + 2 \right) 2^{-\gamma} +
  \frac{3}{\NPPitX} + \frac{3}{\NPPitY} \bigg],
\\ \la^\prime &= \sum_{t=1}^d \delta_t,
\end{align*}
where $\delta_t$ is given by \eqref{e:delta_t}.
\end{remark}

%%%
\begin{proof}
Consider a virtual protocol which does not terminate even if
the total number of bits exceed $l_{\max}$. Denote the
output of this protocol by $\bar{\Prot}_X
=(\bar{\Prot}_{1\cX},\ldots,\bar{\Prot}_{d\cX})$ and
$\bar{\Prot}_Y =
(\bar{\Prot}_{1\cY},\ldots,\bar{\Prot}_{d\cY})$.  We have
\begin{align}
& \ttlvrn{\bPP{\Prot_{\cX} \Prot_{\cY} XY}}{ \bPP{\Prot
      \Prot XY}} \nonumber \\ &\le \ttlvrn{\bPP{\Prot_{\cX}
      \Prot_{\cY} XY}}{ \bPP{\bar{\Prot}_{\cX}
      \bar{\Prot}_{\cY} XY}} + \ttlvrn{
    \bPP{\bar{\Prot}_{\cX} \bar{\Prot}_{\cY} XY}}{
    \bPP{\Prot \Prot XY}} \nonumber \\ &\le \bPr{
    (\Prot_{\cX},\Prot_{\cY}) \neq
    (\bar{\Prot}_{\cX},\bar{\Prot}_{\cY}) } + \ttlvrn{
    \bPP{\bar{\Prot}_{\cX} \bar{\Prot}_{\cY} XY}}{
    \bPP{\Prot \Prot XY}} .
 \label{eq:final-simulation-1} 
\end{align}
First, we bound the second term of
\eqref{eq:final-simulation-1}. By using triangular
inequality repeatedly and by using Lemma
\ref{l:simulation_one_message2}, we have
\begin{align}
& \ttlvrn{ \bPP{\bar{\Prot}_{\cX} \bar{\Prot}_{\cY}
      XY}}{\bPP{\Prot \Prot XY}} \nonumber \\ &\le \ttlvrn{
    \bPP{\bar{\Prot}_{1\cX}\bar{\Prot}_{1\cY}\cdots\bar{\Prot}_{(d-1)\cX}\bar{\Prot}_{(d-1)\cY}\bar{\Prot}_{d\cX}\bar{\Prot}_{d\cY}XY}
  }{\bPP{\Prot_1\Prot_1\cdots\Prot_{(d-1)}\Prot_{(d-1)}\bar{\Prot}_{d\cX}\bar{\Prot}_{d\cY}
      XY}} \nonumber \\ &~~~+
  \ttlvrn{\bPP{\Prot_1\Prot_1\cdots\Prot_{(d-1)}\Prot_{(d-1)}\bar{\Prot}_{d\cX}\bar{\Prot}_{d\cY}
      XY}}{\bPP{\Prot_1\Prot_1\cdots\Prot_{(d-1)}\Prot_{(d-1)}\Prot_{d}\Prot_{d}XY}}
  \nonumber \\ &=
  \ttlvrn{\bPP{\bar{\Prot}_{1\cX}\bar{\Prot}_{1\cY}\cdots\bar{\Prot}_{(d-1)\cX}\bar{\Prot}_{(d-1)\cY}XY}}{\bPP{\Prot_1\Prot_1\cdots\Prot_{(d-1)}\Prot_{(d-1)}
      XY}} \nonumber \\ &~~~+
  \ttlvrn{\bPP{\bar{\Prot}_{d\cX}\bar{\Prot}_{d\cY}(X\Prot^{d-1})(Y\Prot^{d-1})}}{\bPP{\Prot_d
      \Prot_d (X\Prot^{d-1})(Y\Prot^{d-1})}} \nonumber \\ &=
  \nonumber \\ &~\vdots \nonumber \\ &= \sum_{t=1}^d
  \ttlvrn{
    \bPP{\bar{\Prot}_{t\cX}\bar{\Prot}_{t\cY}(X\Prot^{t-1})(Y\Prot^{t-1})}}{
    \bPP{\Prot_t \Prot_t (X\Prot^{t-1})(Y\Prot^{t-1}) }}
  \nonumber \\ &\le \sum_{t: \text{odd}} \bigg[ \bPr{
      (\Prot_t,(Y,\Prot^{t-1})) \in \TPPitYzero } + \bPr{
      (\Prot_t,(X,\Prot^{t-1})) \in \TPPitXzero } \bigg]
  \nonumber \\ &~~~ + \left( \NPPitY + 1 \right) 2^{-\gamma}
  + \frac{1}{\NPPitX} \bigg] \nonumber \\ &~~~ +
    \sum_{t:\text{even}} \bigg[ \bPr{
        (\Prot_t,(Y,\Prot^{t-1})) \in \TPPitYzero } + \bPr{
        (\Prot_t,(X,\Prot^{t-1})) \in \TPPitXzero } \bigg]
    \nonumber \\ &~~~+ \left( \NPPitX + 1 \right)
    2^{-\gamma} + \frac{1}{\NPPitY} \bigg] \nonumber \\ &\le
      \sum_{t=1}^d \bigg[ \bPr{ (\Prot_t,(Y,\Prot^{t-1}))
          \in \TPPitYzero } + \bPr{
          (\Prot_t,(X,\Prot^{t-1})) \in \TPPitXzero }
        \nonumber \\ &~~~+ \left( \NPPitY + \NPPitX + 2
        \right) 2^{-\gamma} + \frac{1}{\NPPitX} +
        \frac{1}{\NPPitY} \bigg].
 \label{eq:final-simulation-2} 
\end{align}
Denote
\begin{align*}
l(X,Y,\bar{\Prot}_{\cX},\bar{\Prot}_{\cY}) &:=
\sum_{t:\text{odd}}
h_{\bPP{\Prot_t|Y\Prot^{t-1}}}(\bar{\Prot}_{t\cX}|Y,\bar{\Prot}_{\cY}^{t-1})
- h_{\bPP{\Prot_t|X\Prot^{t-1}}}(\bar{\Prot}_{t\cX}|X,
\bar{\Prot}_{\cX}^{t-1}) \\ &~~~ +\sum_{t:\text{even}}
h_{\bPP{\Prot_t|X\Prot^{t-1}}}(\bar{\Prot}_{t\cY}|X,\bar{\Prot}_{\cX}^{t-1})
- h_{\bPP{\Prot_t|Y\Prot^{t-1}}}(\bar{\Prot}_{t\cY}|Y,
\bar{\Prot}_{\cY}^{t-1}).
\end{align*}
Since $(\Prot_{\cX},\Prot_{\cY})$ coincides with
$(\bar{\Prot}_{\cX},\bar{\Prot}_{\cY})$ when the accumulated
message length of the protocol generating
$(\bar{\Prot}_{\cX},\bar{\Prot}_{\cY})$ does not exceed
$l_{\max}$, and since the message length of each round is
bounded by each term of
$l(X,Y,\bar{\Prot}_{\cX},\bar{\Prot}_{\cY})$ plus $\delta_t$
by Lemma \ref{l:simulation_one_message2} unless
$(\bar{\Prot}_{t\cX}, (Y,\bar{\Prot}_{\cY}^{t-1})) \in
\TPPitYzero$ or $(\bar{\Prot}_{t\cY}, (X,
\bar{\Prot}_{\cX}^{t-1})) \in \TPPitXzero$, we have
\begin{align}
& \bPr{ (\Prot_{\cX},\Prot_{\cY}) \neq
    (\bar{\Prot}_{\cX},\bar{\Prot}_{\cY}) } \nonumber
  \\ &\le \bPr{ l(X,Y,\bar{\Prot}_{\cX},\bar{\Prot}_{\cY}) +
    \sum_{t=1}^d \delta_t > l_{\max} } \nonumber \\ &~~~ +
  \bPr{ \bigcup_{t:\text{odd}} (\bar{\Prot}_{t\cX},
    (Y,\bar{\Prot}_{\cY}^{t-1})) \in \TPPitYzero \text{ or }
    \bigcup_{t:\text{even}} (\bar{\Prot}_{t\cY}, (X,
    \bar{\Prot}_{\cX}^{t-1})) \in \TPPitXzero}
   \label{eq:final-simulation-3} 
\end{align}
Since
\begin{align*}
\bPr{ (X,Y,\bar{\Prot}_{\cX},\bar{\Prot}_{\cY}) \in \cE }
\le \bPr{ (X,Y,\Prot,\Prot) \in \cE} + \ttlvrn{
  \bPP{\bar{\Prot}_{\cX} \bar{\Prot}_{\cY} XY}}{\bPP{\Prot
    \Prot XY}}
 \label{eq:final-simulation-4} 
\end{align*}
for any event $\cE$, it follows from
\eqref{eq:final-simulation-3} that
\begin{align}
& \bPr{ (\Prot_{\cX},\Prot_{\cY}) \neq
    (\bar{\Prot}_{\cX},\bar{\Prot}_{\cY}) } \nonumber
  \\ &\le \bPr{ l(X,Y,\Prot,\Prot) + \sum_{t=1}^d \delta_t >
    l_{\max} } \nonumber \\ &~~~ + \bPr{
    \bigcup_{t:\text{odd}} (\Prot_{t}, (Y,\Prot^{t-1})) \in
    \TPPitYzero \text{ or } \bigcup_{t:\text{even}}
    (\Prot_{t}, (X, \Prot^{t-1})) \in \TPPitXzero} \nonumber
  \\ &~~~ + 2 \ttlvrn{\bPP{\bar{\Prot}_{\cX}
      \bar{\Prot}_{\cY} XY}}{\bPP{\Prot \Prot XY}} \nonumber
  \\ &\le \bPr{ l(X,Y,\Prot,\Prot) + \sum_{t=1}^d \delta_t >
    l_{\max} } \nonumber \\ &~~~ + \sum_{t=1}^d \bigg[ \bPr{
      (\Prot_t,(Y,\Prot^{t-1})) \in \TPPitYzero } + \bPr{
      (\Prot_t,(X,\Prot^{t-1})) \in \TPPitXzero } \bigg]
  \nonumber \\ & ~~~ + 2 \ttlvrn{\bPP{\bar{\Prot}_{\cX}
      \bar{\Prot}_{\cY} XY}}{ \bPP{\Prot \Prot
      XY}}. \nonumber \\
\end{align}
Thus, by combining this bound with
\eqref{eq:final-simulation-1} and
\eqref{eq:final-simulation-2}, and by noting
\begin{align*}
 l(X,Y,\Prot,\Prot) = \icp,
\end{align*}
we have the desired bound on simulation error.
\end{proof}

%%%%%%%%%%%%%%%%%%%%%%%%%%%%%%%%%%%%%%%%%%%%%%%%%%

%%%%%%%%%%%%%%%%%%%%%%%%%%%%%%%%%%%%%%%%%%%%%%%%%%
\section{Asymptotic Optimality} \label{sec:asymptotic-optimality}
We now present the proofs of
Theorem~\ref{theorem:second-order} and
Theorem~\ref{theorem:general}. Both the proofs rely on
carefully choosing the slice-sizes in the lower and upper
bounds.

%%%%
\subsection{Proof of Theorem \ref{theorem:second-order}}
We start with the upper bound. Note that, for IID random
variables $(\Prot^n,X^n,Y^n)$, the spectrums of $h(\Prot_t^n
| Z^n, (\Prot^{t-1})^n)$ for\footnote{We use this notation
  throughout this section to avoid repetition.} $Z = X
\text{ or } Y$ have width $O(\sqrt{n})$. Therefore, the
parameters $\Delta$s and $N$s that appear in the fudge
parameters can be chosen as $O(n^{1/4})$. Specifically, by
standard measure concentration bounds (for bounded random
variables), for every $\nu > 0$, there exists a
constant\footnote{Although the constant depends on random
  variables appearing in each round, since the number of
  rounds is bounded, we take the maximum constant so that
  \eqref{eq:tail-bound-second-order-proof} holds for every
  $t$.} $c > 0$ such that with
\begin{align*}
\lambda_{\bPP{\Prot_t^n | Z^n (\Prot^{t-1})^n}}^{\min} &= n
H(\Prot_t | Z, \Prot^{t-1}) - c \sqrt{n},
\\ \lambda_{\bPP{\Prot_t^n | Z^n (\Prot^{t-1})^n}}^{\max} &=
n H(\Prot_t | Z, \Prot^{t-1}) + c \sqrt{n},
\end{align*}
the following bound holds:
\begin{align} \label{eq:tail-bound-second-order-proof}
\bPr{ (\Prot_t^n, (Z^n, (\Prot^{t-1})^n )) \in
  \cT_{\bPP{\Prot_t^n | Z^n (\Prot^{t-1})^n}}^{(0)} } \le
\nu.
\end{align}
Let $T$ denote the third central moment of the random
variable $\icp$. For
\begin{align*}
\lambda_n = n \Icp + \sqrt{n \mathtt{V}(\prot)} Q^{-1}\left(
\ep - 9d \nu - \frac{T^3}{2 \mathtt{V}(\prot)^{3/2}
  \sqrt{n}} \right),
\end{align*}
choosing $\Delta_{\bPP{\Prot_t^n | Z^n (\Prot^{t-1})^n}} =
N_{\bPP{\Prot_t^n | Z^n (\Prot^{t-1})^n}} = \gamma =
\sqrt{2c} n^{1/4}$, and $l_{\max} = \lambda_n + \sum_{t=1}^d
\delta_t$ in Theorem \ref{t:upper_bound} (for the definition
of the fudge parameters, see Remark
\ref{r:fudge_upper_bound}), we get a protocol of length
$l_{\max}$ and satisfying
\begin{align*}
\ttlvrn{\bPP{\Prot_\cX^n \Prot_\cY^n X^n Y^n}}{\bPP{\Prot^n
    \Prot^n X^n Y^n}} \le \bPr{ \sum_{i=1}^n
  \mathtt{ic}(\Prot_i; X_i,Y_i) > \lambda_n } + 9d\nu
\end{align*}
for sufficiently large $n$. By its definition given in \eqref{e:delta_t},
$\delta_t = O(n^{1/4})$ for the choice of parameters above.  Thus, the Berry-Ess\'een theorem
(\cf~\cite{Fel71}) and the observation above gives a
protocol of length $l_{\max}$ attaining $\ep$-simulation.
Therefore, using the Taylor approximation of $Q(\cdot)$
yields the achievability of the claimed protocol length.

For the lower bound, we fix sufficiently small constant
$\delta > 0$, and we set $\lamins1 = n( H(X,Y) - \delta)$,
$\lamaxs1 = n(H(X,Y) + \delta)$, $\lamins2 = n (H(X|Y,\Prot)
- \delta)$, $\lamaxs2 = n( H(X|Y,\Prot) + \delta)$,
$\lamins3 = n (H(X\Prot \triangle Y \Prot) -\delta)$,
$\lamaxs3 = n (H(X\Prot \triangle Y \Prot) +\delta)$,
respectively.  Then, by standard measure concentration
bounds imply that the tail probability $\ep_{\mathtt{tail}}$
in \eqref{e:tail_prob} is bounded above by $\frac{c}{n}$ for
some constant $c > 0$. We also set $\eta = \frac{1}{n}$.
For these choices of parameters, we note that the fudge
parameter is $\lambda^\prime = O(\log n)$.  Thus, by setting
\begin{align*}
\lambda = \lambda_n &= n \Icp + \sqrt{n \mathtt{V}(\prot)}
Q^{-1}\left( \ep + \frac{c+2}{n} + \frac{T^3}{2
  \mathtt{V}(\prot)^{3/2}\sqrt{n}} \right) \\ &= n \Icp +
\sqrt{n \mathtt{V}(\prot)} Q^{-1}(\ep) + O(\log n),
\end{align*}
where the final equality is by the Tailor approximation, an
application of the Berry-Ess\'een theorem to the bound in
\eqref{e:lower_bound} gives the desired lower bound on the
protocol length. \qed

%%%%%%%%%%%%%%%%%%%%%
\subsection{Proof of Theorem \ref{theorem:direct-product}}
Theorem \ref{t:lower_bound} implies that if a protocol
$\psim$ is such that
\begin{align} \label{eq:proof-direct-product-1}
\log | \psim | < \la- \lasmall,
\end{align}
then its simulation error must be larger than
\begin{align} \label{eq:proof-direct-product-2}
\bPr{ \mathtt{ic}\left(\Prot^n; X^n,Y^n \right) > \la} -
\epsmall.
\end{align}
To compute fudge parameters, we set $\lamins1 = n( H(X,Y) -
\delta)$, $\lamaxs1 = n(H(X,Y) + \delta)$, $\lamins2 = n
(H(X|Y,\Prot) - \delta)$, $\lamaxs2 = n( H(X|Y,\Prot) +
\delta)$, $\lamins3 = n (H(X\Prot \triangle Y \Prot)
-\delta)$, $\lamaxs3 = n (H(X\Prot \triangle Y \Prot)
+\delta)$, respectively.  By the Chernoff bound, there
exists $E_1 > 0$ such that
\begin{align*}
\ep_{\mathtt{tail}} \le 2^{-E_1 n}.
\end{align*}
Furthermore, $\Lambda_i = O(n)$ for $i=1,2,3$. We set $\eta
= 2^{- \frac{\delta}{27}n}$.  It follows that
\begin{align} \label{eq:proof-direct-product-3}
\epsmall \le 2^{- E_1 n} + 2^{- \frac{\delta}{27}n}
\end{align}
and
\begin{align} \label{eq:proof-direct-product-4}
\lasmall \le \frac{\delta}{3} n + O(\log n).
\end{align}
Finally, upon setting
\begin{align} \label{eq:proof-direct-product-5}
\la = n \Icp - \frac{\delta}{3}
\end{align}
and applying the Chernoff bound once more, we obtain a
constant $E_2 > 0$ such that
\begin{align} \label{eq:proof-direct-product-6}
\bPr{ \mathtt{ic}\left(\Prot^n; X^n,Y^n \right) > \la} \ge 1
- 2^{- E_2 n}.
\end{align}
The result follows upon combining
\eqref{eq:proof-direct-product-1}-\eqref{eq:proof-direct-product-6}. \qed

%%%%
\subsection{Proof of Theorem \ref{theorem:general}}

For a sequence of protocols $\boldsymbol{\prot} = \{ \prot_n
\}_{n=1}^\infty$ and a sequence of observations
$(\mathbf{X},\mathbf{Y}) = \{ (X_n,Y_n) \}_{n=1}^\infty$,
let
\begin{align}
\underline{H}(\boldsymbol{\Prot}_t |
\mathbf{Z},\boldsymbol{\Prot}^{t-1}) &= \sup\left\{ \alpha :
\lim_{n\to\infty} \bPr{ h(\Prot_{n,t}|Z_n\Prot_n^{t-1}) <
  \alpha } = 0 \right\},
  \label{eq:conditional-infH} \\
\overline{H}(\boldsymbol{\Prot}_t |
\mathbf{Z},\boldsymbol{\Prot}^{t-1}) &= \inf\left\{ \alpha :
\lim_{n\to\infty} \bPr{ h(\Prot_{n,t}|Z_n\Prot_n^{t-1}) >
  \alpha } = 0 \right\},
  \label{eq:conditional-supH} 
\end{align}
where $\mathbf{Z} = \mathbf{X} \text{ or } \mathbf{Y}$,
$\boldsymbol{\Prot}_t = \{ \Prot_{n,t} \}_{n=1}^\infty$ and
$\boldsymbol{\Prot}_n^{t-1} = \{ \Prot_n^{t-1}
\}_{n=1}^\infty$ are sequences of transcripts of $t$th round
and up to $t$th rounds, respectively.  For achievability
part, we fix arbitrary small $\delta > 0$, and set
\begin{align*}
\lambda_{\bPP{\Prot_{n,t}|Z_n \Prot_n^{t-1}}}^{\min} &= n
\left( \underline{H}(\boldsymbol{\Prot}_t |
\mathbf{Z},\boldsymbol{\Prot}^{t-1}) - \delta \right),
\\ \lambda_{\bPP{\Prot_{n,t}|Z_n \Prot_n^{t-1}}}^{\max} &= n
\left( \overline{H}(\boldsymbol{\Prot}_t |
\mathbf{Z},\boldsymbol{\Prot}^{t-1}) + \delta \right),
\end{align*}
$\Delta_{\bPP{\Prot_{n,t}|Z_n \Prot_n^{t-1}}} =
N_{\bPP{\Prot_{n,t}|Z_n \Prot_n^{t-1}}} = \gamma =
\sqrt{2\delta n}$.  We set
\begin{align*}
l_{\max} &= n
\left(\overline{\mathtt{IC}}(\boldsymbol{\prot}) + \delta
\right) + \sum_{t=1}^d\delta_t \\ &= n
\left(\overline{\mathtt{IC}}(\boldsymbol{\prot}) + \delta
\right) + O(\sqrt{n}),
\end{align*}
where $\delta_t$ is given by \eqref{e:delta_t}.
Then, by Theorem \ref{t:upper_bound}, by the definition of
$\overline{\mathtt{IC}}(\boldsymbol{\prot})$ and by
\eqref{eq:conditional-infH} and \eqref{eq:conditional-supH},
there exists a simulation protocol of length $l_{\max}$ with
vanishing simulation error. Since $\delta > 0$ is arbitrary,
we have the desired achievability bound.

For converse part, we fix arbitrary $\delta > 0$, and set
$\lamins1 = n( \underline{H}(\mathbf{X},\mathbf{Y}) -
\delta)$, $\lamaxs1 = n(\overline{H}(\mathbf{X},\mathbf{Y})
+ \delta)$, $\lamins2 = n
(\underline{H}(\mathbf{X}|\mathbf{Y},\boldsymbol{\Prot}) -
\delta)$, $\lamaxs2 = n(
\overline{H}(\mathbf{X}|\mathbf{Y},\boldsymbol{\Prot}) +
\delta)$, $\lamins3 = n (\underline{H}(\mathbf{X}
\boldsymbol{\Prot} \triangle \mathbf{Y} \boldsymbol{\Prot})
-\delta)$, $\lamaxs3 = n (\overline{H}(\mathbf{X}
\boldsymbol{\Prot} \triangle \mathbf{Y} \boldsymbol{\Prot})
+\delta)$, respectively, where
\begin{align*}
\underline{H}(\mathbf{X},\mathbf{Y}) &= \sup\left\{ \alpha :
\lim_{n\to\infty} \bPr{ h(X_n Y_n) < \alpha } = 0 \right\},
\\ \overline{H}(\mathbf{X},\mathbf{Y}) &= \inf\left\{ \alpha
: \lim_{n\to\infty} \bPr{ h(X_n Y_n) > \alpha } = 0
\right\},
\\ \underline{H}(\mathbf{X}|\mathbf{Y},\boldsymbol{\Prot})
&= \sup\left\{ \alpha : \bPr{ h(X_n|Y_n \Prot_n) < \alpha }
= 0 \right\},
\\ \overline{H}(\mathbf{X}|\mathbf{Y},\boldsymbol{\Prot}) &=
\inf\left\{ \alpha : \bPr{ h(X_n|Y_n \Prot_n) > \alpha } = 0
\right\}, \\ \underline{H}(\mathbf{X} \boldsymbol{\Prot}
\triangle \mathbf{Y} \boldsymbol{\Prot}) &= \sup\left\{
\alpha : \bPr{ - h( X_n \Prot_n \triangle Y_n \Prot_n) <
  \alpha } = 0 \right\}, \\ \overline{H}(\mathbf{X}
\boldsymbol{\Prot} \triangle \mathbf{Y} \boldsymbol{\Prot})
&= \inf\left\{ \alpha : \bPr{ - h( X_n \Prot_n \triangle Y_n
  \Prot_n) > \alpha } = 0 \right\}.
\end{align*} 
Then, by the definitions, we find that the tail probability
$\ep_{\mathtt{tail}}$ in \eqref{e:tail_prob} converges to
$0$. We also set $\eta = (1/n)$. For these choices of
parameters, we note that the fudge parameter is
$\lambda^\prime = O(\log n)$. Thus, by using the bound in
\eqref{e:lower_bound} for
\begin{align}
\lambda = \lambda_n = n \left(
\overline{\mathtt{IC}}(\boldsymbol{\prot}) + \delta \right),
\end{align}
and by taking $\delta \to 0$, we have the desired converse
bound.  \qed

%%%%%%%%%%%%%%%%%%%%%%%%%%%%%%%%%%%%%%%%%%%%
\section{Conclusion} \label{s:conclusion}
We have proposed a {\it common randomness decomposition} based
approach ($cf.$~\cite{TyaThesis}) to derive a lower bound on communication complexity
of protocol simulation by relating the protocol simulation problem to the secret key agreement.
A key step in our approach is identifying the amount of common randomness generated through protocol simulation. 
Our estimate for the amount of common randomness does not rely on the structure of the function to be computed. 
This is contrast to most of the existing lower bounds on communication complexity for function computation, 
such as the partition bound or the discrepancy bound, where the structure of the computed function plays an important role. 
In particular, a comparison of our approach with other existing approaches for specific functions is not available. An important 
future  research agenda for us is to incorporate the structure of functions in our bound; the case of functions with a small range
such as Boolean functions is of particular interest.

%%%%%%%%%%%%%%%%%%%%%%%%%%%%%%%%%%%%%%%%%%%%
%\newpage
\appendix
\section{Example Protocol} \label{appendix:small-error}
To illustrate the utility of our lower bound, we consider a
protocol $\prot$ which takes very few values most of the
time, but with very small probability it can send many
different transcripts.  The proposed protocol can be
$\ep$-simulated using very few bits of communication on
average. But in the worst-case it requires as many bits of
communication for $\ep$-simulation as needed for data
exchange, for all $\ep>0$ small enough.

Specifically, let $\cX = \cY = \{1,\ldots,2^n\}$ and let
$\prot$ be a deterministic protocol such that the transcript
$\tau(x,y)$ for $(x,y)$ is given by
\begin{align*}
\tau(x,y) = \left\{
\begin{array}{cl}
a & \text{if } x > \delta 2^n, y > \delta 2^n \\ b &
\text{if } x > \delta 2^n, y \le \delta 2^n \\ c & \text{if
} x \le \delta 2^n, y > \delta 2^n \\ (x,y) & \text{if } x
\le \delta 2^n, y \le \delta 2^n
\end{array}
\right.
\end{align*}
for some small $\delta > 0$, which will be specified
later. Clearly, this protocol is interactive.

Let $(X,Y)$ be the uniform random variables on $\cX \times
\cY$. Then,
\begin{align*}
\bPr{ \Prot \notin \{a,b,c\}} = \delta^2.
\end{align*}
Since
\begin{align*}
\bPP{\Prot|X}(\tau(x,y)|x) = \left\{
\begin{array}{cl}
1-\delta & \text{if } x > \delta 2^n, y > \delta 2^n
\\ \delta & \text{if } x > \delta 2^n, y \le \delta 2^n
\\ 1-\delta & \text{if } x \le \delta 2^n, y > \delta 2^n
\\ \frac{1}{2^n} & \text{if } x \le \delta 2^n, y \le \delta
2^n
\end{array}
\right.
\end{align*}
and similarly for $\bPP{\Prot|Y}(\tau(x,y)|y)$, we have
\begin{align*}
\mathtt{ic}(\tau(x,y);x,y) = \left\{
\begin{array}{cl}
2 \log (1/(1-\delta)) & \text{if } x > \delta 2^n, y >
\delta 2^n \\ \log (1/\delta) + \log (1/(1-\delta)) &
\text{if } x > \delta 2^n, y \le \delta 2^n \\ \log
(1/\delta) + \log(1/(1-\delta)) & \text{if } x \le \delta
2^n, y > \delta 2^n \\ 2n & \text{if } x \le \delta 2^n, y
\le \delta 2^n
\end{array}
\right..
\end{align*}
Consider $\delta = \frac{1}{n}$, and $\ep = \frac{1}{n^3}$.
Note that for any $\la < 2n$,
\[
\bPr{\icp > \la} \geq \bPr{\Prot\{a,b,c\}} = \delta^2 =
\frac 1 {n^2} > \ep,
\]
and
\[
\bPr{\icp > 2n} =0.
\]
Thus, the $\ep$-tail of information complexity density
$\la_\ep = \sup\{\la: \bPr{\icp > \la} > \ep\}$ is given by
\begin{align}
\lambda_\ep = 2n.
\label{e:la_bound}
\end{align}
On the other hand, we have
\begin{align*}
\Icp &= H(\Prot|X) + H(\Prot|Y) \\ &\le 2 \delta [ h_b(\delta)
  + \log n - \log (1/\delta)] + 2 (1-\delta) h_b(\delta)
\\ &\leq \tOrder(\delta^2)
\end{align*}
where $h_b(\cdot)$ is the binary entropy function.

Also, to evaluate the lower bound of
Theorem~\ref{t:lower_bound}, we bound the fudge parameters
in that bound. To that end, we fix $\eptail = 0$ and bound
the spectrum lengths $\Lambda_1, \Lambda_2, \Lambda_3$.
Since $(X,Y)$ is uniform, $h(X,Y) = 2n$ and so, $\Lambda_1
=0$. Also, note that with probability $1$ the conditional
entropy density $h(X|\Prot,Y)$ is either $0$ or $\log
(\delta 2^n)$, which implies $\Lambda_2 = \Order(n)$. A
similar argument shows that $\Lambda_3 =
\Order(n)$. Therefore, the fudge parameter
\[
\lasmall = \Order(\log \Lambda_1\Lambda_2\Lambda_3) =
\Order(\log n),
\]
which in view of \eqref{e:la_bound} and
Theorem~\ref{t:lower_bound} gives $\dcp = \Omega(2n)$.  \qed
%%%%%%%%%%%%%%%%%%%%%%%%%%%%%%%%%%%%%%%%
\section{Proof of Lemma~\ref{l:SK_relation}}\label{a:proof_lemma}

\begin{lemma*}
Consider random variables $X, Y, Z$ and $V$ taking values in
countable sets $\cX$, $\cY$, $\cZ$, and a finite set $\cV$,
respectively.  Then, for every $0< \ep < 1/2$,
\[
S_{2\ep}(X,Y|ZV) \geq S_\ep(X,Y|Z) - \log|\cV| -
2\log(1/2\ep).
\]
\end{lemma*}
{\it Proof.} Consider random variables $K_\cX^\prime$ and
$K_\cY^\prime$ with a common range $\cK^\prime$ such that
$(K_\cX^\prime, K_\cY^\prime)$ constitutes an $\ep$-secret
key for $X$ and $Y$ given eavesdropper's observation $Z$,
recoverable using an interactive protocol $\prot^\prime$.
Let $\bQQ{K_\cX^\prime K_\cY^\prime\Prot^\prime ZV}$ denote
the distribution $\bPP{\mathtt{unif}}^{\prime(2)}
\bPP{\Prot^\prime ZV}$, where
$\bPP{\mathtt{unif}}^{\prime(2)}$ denotes the distribution
\[
\bPP{\mathtt{unif}}^{\prime(2)}(k_\cX, k_\cY) = \frac
    {\indicator(k_\cX = k_\cY)} {|\cK^\prime|}, \quad
    \forall\, k_\cY, k_\cY \in \cK^\prime.
\]
Then, by definition of an $\ep$-secret key, it holds that
\begin{align}
\ttlvrn{\bPP{K_\cX^\prime K_\cY^\prime \Prot^\prime
    Z}}{\bQQ{K_\cX^\prime K_\cY^\prime \Prot^\prime Z}} \leq
\ep.
\label{e:SK_condition}
\end{align}
Note that $H_{\min}(\bQQ{K^\prime_\cX \Prot^\prime Z}\mid
\Prot^\prime Z) \geq \log|\cK^\prime|$.  Therefore, by
Lemma~\ref{l:leftover_hash} there exists a function $K_\cX =
K(K^\prime_\cX)$ taking values in a set $\cK$ with
$\log|\cK| \geq \log|\cK^\prime| - \log|\cV| -
2\log(1/2\ep)$ such that
\begin{align}
\ttlvrn{\bQQ{K_\cX \Prot^\prime Z
    V}}{\bPP{\mathtt{unif}}\bQQ{\Prot^\prime Z V}} \leq \ep,
\label{e:SK_construction}
\end{align}
where $\bPP{\mathtt{unif}}$ denotes the uniform distribution
on the set $\cK$. Upon letting $K_\cY = K(K_\cY^\prime)$ and
defining $\bPP{\mathtt{unif}}^{(2)}$ analogously to
$\bPP{\mathtt{unif}}^{\prime(2)}$ with $\cK$ in place of
$\cK^\prime$, we have
\begin{align*}
\ttlvrn{\bPP{K_\cX K_\cY\Prot^\prime Z
    V}}{\bPP{\mathtt{unif}}^{(2)}\bPP{\Prot^\prime Z V}}
&\leq \ttlvrn{\bQQ{K_\cX K_\cY\Prot^\prime Z
    V}}{\bPP{\mathtt{unif}}^{(2)}\bPP{\Prot^\prime Z V}} +
\ep \\ &= \ttlvrn{\bQQ{K\Prot^\prime
    ZV}}{\bPP{\mathtt{unif}}\bPP{\Prot^\prime Z V}} + \ep
\\ &\leq 2\ep,
\end{align*}
where the first inequality is by \eqref{e:SK_condition} and
the second by \eqref{e:SK_construction}, and the equality is
by the definition of $\dQ$. Therefore, $(K_\cX, K_\cY)$
constitutes a $2\ep$-secret key of length $\log|\cK^\prime|
- \log|\cV| - 2\log(1/2\ep)$ for $X$ and $Y$ given
eavesdropper's observation $(Z,V)$. The claimed bound
follows since $K^\prime$ was an arbitrary secret key for $X$
and $Y$ given eavesdropper's observation $Z$.  \qed

%%%%%%%%%%%%%%%%%%%%%%%%%%%%%%%%%%%%%%%%%%%%%%%
\bibliography{IEEEabrv,references}
\bibliographystyle{IEEEtranS}

 %%%%%%%%%%%%%%%%%%%%%%%%%%%%%%%%%%%%%%%%%
\end{document}